 \documentclass[final,1p,times]{elsarticle}
\usepackage{multirow}
\usepackage{cite}
\usepackage{graphicx}
\usepackage{subfigure}
\usepackage{amsmath}
\usepackage{amsfonts}
\usepackage{algorithm}
\usepackage{algorithmic}
\usepackage{bm}
\usepackage{picins}
\usepackage{url}

\newtheorem{thm}{Theorem}[section]

\newtheorem{lem}[thm]{Lemma}
\newtheorem{definition}[thm]{Definition}
\journal{....}

\begin{document}

\begin{frontmatter}

%% Title, authors and addresses

%% use the tnoteref command within \title for footnotes;
%% use the tnotetext command for the associated footnote;
%% use the fnref command within \author or \address for footnotes;
%% use the fntext command for the associated footnote;
%% use the corref command within \author for corresponding author footnotes;
%% use the cortext command for the associated footnote;
%% use the ead command for the email address,
%% and the form \ead[url] for the home page:
%%
%% \title{Title\tnoteref{label1}}
%% \tnotetext[label1]{}
%% \author{Name\corref{cor1}\fnref{label2}}
%% \ead{email address}
%% \ead[url]{home page}
%% \fntext[label2]{}
%% \cortext[cor1]{}
%% \address{Address\fnref{label3}}
%% \fntext[label3]{}

\title{Breaching Euclidean Distance-Preserving Data Perturbation Using Few Known Inputs}

%% use optional labels to link authors explicitly to addresses:
%% \author[label1,label2]{Chris R. Giannella}
%% \address[label1]{The MITRE Corporation // 300 Sentinel Dr. Suite 600 // Annapolis Junction, MD 20701}
%% \address[label2]{<address>}

\author{Chris R. Giannella\corref{cor1}}\corref{cor2}
\ead{cgiannella@mitre.org}
\cortext[cor1]{Giannella is the corresponding author.  This document was approved for public release, unlimited distribution, by The MITRE Corporation under case 10-2893.  Giannella's affiliation with The MITRE Corporation is provided for identification purposes only, and is
not intended to convey or imply MITRE concurrence with, or support of, the positions, opinions or viewpoints expressed.}
\address{The MITRE Corporation, 300 Sentinel Dr. Suite 600, Annapolis Junction MD 20701, (301) 617-3000}
\author{Kun Liu}
\ead{kun@linkedin.com}
\address{LinkedIn, 2029 Stierlin Court, Mountain View, CA 94043}
\author{Hillol Kargupta\corref{cor2}}
\ead{hillol@cs.umbc.edu}
\cortext[cor2]{Kargupta is also affiliated with AGNIK LLC, Columbia MD}
\address{Dept. of CSEE, University of Maryland Baltimore County, Baltimore MD 21250}

%%\author{Double Blind -- No Author Name}
%%\address{No Address}
%%\ead{No Email}

\begin{abstract}
We examine Euclidean distance-preserving data perturbation as a
tool for privacy-preserving data mining. Such perturbations allow many
important data mining algorithms ({\em e.g.} hierarchical and k-means clustering), with only minor modification, to be applied
to the perturbed data and produce exactly the same results as if applied to the original
data.
However, the issue of how well the privacy of the original data is preserved needs
careful study.  We engage in this study by assuming the
role of an attacker armed with a small set of known original data tuples (inputs).  Little work
has been done examining this kind of attack when the number of known original tuples is less
than the number of data dimensions.  We focus on this important case, develop and rigorously
analyze an attack that utilizes {\em any number} of known original tuples.  The approach allows
the attacker to estimate the original data tuple associated with each perturbed tuple and
calculate the probability that the estimation
results in a privacy breach.  On a real 16-dimensional dataset, we show that the attacker, with 4 known original tuples, 
can estimate an original unknown tuple with less than 7\% error with probability exceeding 0.8.
\end{abstract}


\end{frontmatter}

%%
%% Start line numbering here if you want
%%
% \linenumbers

\section{Introduction} \label{introduction}
Owners of sensitive information face a dilemma in many situations.
On the one hand, making this data available for statistical analysis can violate
the privacy of the individuals represented in the data or reveal sensitive information
about the data owner.  On the other hand, making the data available can lead to discoveries 
that provide societal benefits. For example, mining health-care data for security/fraud issues
may require analyzing clinical records and
pharmacy transaction data of many individuals over a certain area.
While the release of such data may violate privacy laws, mining it
can improve the overall quality of the health-care system.  Privacy-Preserving Data
Mining (PPDM) strives to provide a solution to this dilemma. It
aims to allow useful data patterns to be extracted without
compromising privacy.

Data perturbation represents one common approach in PPDM. Here,
the original private dataset $X$ is perturbed and the resulting
dataset $Y$ is released for analysis. Perturbation approaches typically face a ``privacy/accuracy''
trade-off. On the one hand, perturbation must not allow the
original data records to be adequately recovered. On the other hand, it
must allow ``patterns'' that hold in the original data to be recovered.  In
many cases, increased privacy comes at the cost of reduced
accuracy and vice versa.  For example, Agrawal and Srikant
\citep{Agrawal:00} proposed adding randomly generated {\em i.i.d.}
noise to the dataset.  They showed how the distribution from which
the original data arose can be estimated using only the perturbed
data and the distribution of the noise.  However, Kargupta {\em et al.} \citep{Kargupta:03a} and
Huang {\em et al.} \citep{Huang_05} pointed out how, in many cases,
the noise can be filtered off leaving a reasonably good estimation
of the original data (further investigated by Guo {\em et al.} \citep{GWL:2008}).
These results point to the fact that unless
the variance of the additive noise is sufficiently large, original
data records can be recovered unacceptably well. However, this
increase in variance reduces the accuracy with which the original
data distribution can be estimated. This privacy/accuracy
trade-off is not limited to additive noise; some other
data transformation techniques suffer from a similar problem, {\em e.g.} k-anonymity \citep{Sweeney:02}.

Recently, Euclidean distance-preserving data perturbation
for the \textit{census model} \footnote{The census model is widely studied in the field
of security control for statistical databases \citep{Adam:89}.}has
gained attention (\citep{Chen_05, Oliveira:03, Oliveira_04p, Liu_06b, MCG:2006, Chen_07t, TPSSL:2008, WCKM:2009})
because it mitigates the privacy/accuracy trade-off by
guaranteeing perfect accuracy.  The census model using Euclidean distance-preserving 
data perturbation can be illustrated as follows.  An organization has a
private, real-valued dataset $X$ (represented as a matrix where each column is
a data record)
and wishes to make it publicly available for data analysis while
keeping the individual records (columns) private. To accomplish
this, $Y = T(X)$ is released to the public where $T(.)$ is a
function, known only to the data owner that preserves Euclidean
distances between columns. With this nice property, many useful data mining
algorithms, with only minor modification, can be applied to $Y$
and produce {\em exactly the same} patterns that would be extracted if
the algorithm was applied directly to $X$.  For example, assume
single-link, agglomerative hierarchical clustering (using
Euclidean distance) is applied directly to $Y$ \citep{HK:2006}.  The cluster memberships in the
resulting dendrogram
will be identical to those in the dendrogram produced if the same
algorithm is applied to $X$.

However, the issue of how well the private data is
hidden after Euclidean distance-preserving data perturbation needs careful
study.  Without any prior knowledge, the
attacker can do very little (if anything) to accurately recover the
private data.  However, no prior knowledge seems an unreasonable
assumption in many situations.  Consideration of prior knowledge-based 
attack techniques against Euclidean distance-preserving
transformations is an important avenue of study.  In this paper, we engage in this
study by considering {\em known input} prior knowledge wherein
the attacker knows a small set of original data tuples (inputs), but does not know their
associated perturbed data tuples.
As pointed out in \citep{TPSSL:2008, WCKM:2009},
this knowledge could be obtained through insider
information.  For example, consider a dataset
where each record corresponds to information about an individual ({\em e.g.}
medical data, census data).  It is reasonable to assume that the
individuals know (1) that a record for themselves appears in the dataset, and (2) the
attributes of the dataset.  As such, each individual knows one record in the original
dataset.  A small group of malicious individuals could then combine their insider
information to produce a larger set of known original data tuples.

{\bf Summary of our contributions:} The goal of the attacker is to use the perturbed data tuples and
known original data tuples to
produce good estimates of {\em unknown} original data tuples along with links to their perturbed counterparts.
To achieve this, we develop an attack technique called the
{\em known input attack} which
proceeds in three steps.

\begin{enumerate}
\item The attacker links as many of the known
original data tuples (columns in $X$) to their corresponding perturbed counterparts (columns in $Y$).
\item For each unlinked perturbed data tuple, the attacker computes the breach probability of the associated unknown
original data tuple.  This is the probability that the following
stochastic procedure will result in an accurate enough estimate of the associated unknown original data tuple
to be considered a privacy
breach (the probability calculation is done by applying a closed-form expression we derive later).
\begin{enumerate}
\item A Euclidean distance-preserving transformation is uniformly chosen from
the space of such transformations that satisfy the original-perturbed (input-output) constraints from step 1.
\item The inverse
of the chosen transformation is used to estimate original data tuples from their perturbed counterparts.
\end{enumerate}
\item The attacker chooses the perturbed data tuples which are most vulnerable to breach based their probabilities
from step 2, {\em e.g.} chooses the one with the maximum probability or chooses all whose probability
exceeds a threshold, and generates estimates of their associated known original data tuples.
\end{enumerate}

When the number of linked, linearly independent known original data tuples exceeds the number of data dimensions, the
privacy breach probability, for all unknown original data tuples, equals one as the estimates are guaranteed to be error-free.  
However, to our knowledge, little work
has been done for the case where the number of known original data tuples is less than the number of data dimensions.
This is an important case, since obtaining original data tuples is likely difficult.  The attacker ought to be able to
utilize however as many as she can get.  Our results demonstrate how the attacker can do this and with increasing
probability of success with respect to the number original data tuples obtained.  Experiments
on real and synthetic data show that even with the number of known original data tuples significantly smaller than the
number of data dimensions, privacy can be breached with high probability.  For example,
on a real 16-dimensional dataset, we show that the attacker can use 4 known original data tuples to estimate an unknown original tuple 
with less than 7\% error with probability exceeding 0.8.

{\bf Paper organization:}  Section \ref{relatedwork} describes related work in data perturbation for privacy-preserving
data analysis.   Section \ref{background} discusses some background material - the definition of $T$, a Euclidean distance-preserving 
data perturbation, and the definition of a privacy breach.  Section \ref{one_to_one} describes the main
contribution of the paper - the known input attack outlined above.  Section \ref{sec:experiments} discusses the
results of experiments on real and synthetic data to evaluate the behavior of the attack.  Section \ref{sec:discussion}
provides a brief summary of the paper and a pointer to an idea for future work.  Proofs and some detailed derivations are
included in an appendix.

\section{Related Work} \label{relatedwork}

In this section, we give a brief overview of a wide variety of data-perturbation techniques. We first introduce methods
that do not preserve Euclidean distance between data tuples. Then we focus on research most relevant to this paper, a
majority of which aim to preserve Euclidean distance by projecting private data to a new space.

\subsection{General Data Perturbation and Transformation Methods}
\noindent \textbf{Additive perturbation:} Adding {\em i.i.d.}
white noise to protect data privacy is one common approach for
statistical disclosure control \citep{Adam:89}. The perturbed data
allows the retrieval of aggregate statistics of the original data
({\em e.g.} sample mean and variance) without disclosing values of
individual records.  Moreover, additive white noise perturbation
has received attention in the data mining literature~\citep{Agrawal:00, Kargupta:03a, Huang_05, GWL:2008}.  Clearly,
additive noise does not preserve Euclidean distance and, therefore, is
fundamentally different than the data perturbation we consider.
An interesting example along these lines is given by Mukherjee {\em et al.} \citep{MBCG:2008}.  They considered additive
noise to the most dominate principal components of the dataset along with a
modification of k-nearest-neighbor classification~\citep{hastie11elements} on the perturbed data to improve
accuracy.  Moreover, they nicely extend to additive noise the
$\rho_1$-to-$\rho_2$ privacy breach measure originally introduced for categorical
data in \citep{Evfimievski:03}.  Another example is Liu {\em et al.} \citep{LKT:2008}.  They argued that the level of additive
noise ought to be flexible per record.  They developed a modified addative noise approach allowing the level of noise to be
varied per record based on data owner preference.

\noindent \textbf{Multiplicative perturbation:} Two traditional
multiplicative data perturbation schemes were studied in the
statistics community \citep{Kim_03}. One scheme multiplies each data
element by a random number that has a truncated Gaussian
distribution with mean one and small variance. The other takes a
logarithmic transformation of the data first, adds multivariate
Gaussian noise, then takes the exponential function
\textit{exp(.)} of the noise-added data. These perturbations allow
summary statistics ({\em e.g.}, mean, variance) of the
attributes to be estimated, but do not preserve Euclidean distances
among records.

To assess the security of traditional multiplicative perturbation
together with additive perturbation, Trottini {\em et al.}
\citep{Trottini_0disclosure} proposed a Bayesian intruder model
that considers both prior and posterior knowledge of the data.
Their overall strategy of attacking the privacy of perturbed data
using prior knowledge is the same as ours.  However, they
particularly focused on linkage privacy breaches, where an
intruder tries to identify the identity (of a person) linked to a
specific record; while we are primarily interested in data record
recovery.  Moreover, they did not consider Euclidean distance-preserving 
perturbation as we do.

\noindent \textbf{k-anonymization:} Samarati and Sweeney
\citep{Sweeney:02, Samarati:01} originally developed the {\em k-anonymity}
model to transform person-specific data. Their work shows that an attacker
can link a subset of data attributes (called quasi-identifiers) with
third-party information to uniquely identify a person even when his
personally identifiable information is not present in the original data.
To mitigate the risk, the authors proposed the suppression or generalization
of values of these quasi-identifiers so that any records in the database,
when projected onto the quasi-identifiers, cannot be distinguished
from at least \emph{k-1} others. This model has drawn much of attention
because of its simple privacy definition. Since its initial appearance,
a variety of extensions have been developed to anonymize transactional
data~\citep{lefevre06workload}, sequential data~\citep{wang06anonymizing},
and mobility data~\citep{abul08never}. We refer interested readers to
the survey book~\citep{aggarwal08privacy} for more details.
It should be noted that none of these approaches consider Euclidean
distance-preserving perturbation as we do.

\noindent \textbf{Data micro-aggregation:} Two
multivariate micro-aggregation approaches have been proposed by
researchers in the data mining area. The technique presented by
Aggarwal and Yu \citep{Aggarwal_0403} partitions the original data
into multiple groups of predefined size.  For each group, a
certain level of statistical information ({\em e.g.}, mean and
covariance) is maintained. This statistical information is used to
create anonymized data that has similar statistical
characteristics to the original dataset. Li \textit{et al.}
\citep{Li_06tree} proposed a kd-tree based perturbation method,
which recursively partitions a dataset into subsets which are
progressively more homogeneous after each partition. The private
data in each subset is then perturbed using the subset average.
The relationships between attributes are argued to be preserved
reasonably well. However, neither of these two approaches preserves
Euclidean distance between the original data tuples.

\noindent \textbf{Data swapping and shuffling:} Data swapping transforms a
database by exchanging values of sensitive attributes among individual records.
Records are exchanged in such a way that the lower-order frequency counts or
marginals are maintained. A variety of refinements and applications of data swapping
have been addressed since its initial appearance. We refer readers
to~\citep{Fienberg_03} for a thorough treatment.  Data shuffling
\citep{MS:2006} is similar to swapping, but is argued to improve on many
of the shortcomings of swapping for numeric data.  However, neither swapping
or shuffling preserves Euclidean distance, which is the focus of this
paper.

\noindent \textbf{Other techniques:} Evfimievski
{\em et al.} \citep{Evfimievski:03}, Rizvi and Haritza
\citep{Rizvi:02} considered the use of categorical data
perturbation in the context of association rule mining. Their
algorithms delete real items and add bogus items to the original
records. Association rules present in the original data can be
estimated from the perturbed data. Along a related line, Verykios
{\em et al.} \citep{Verykios_03} considered perturbation techniques
which allow the discovery of some association rules while hiding
others considered to be sensitive. We refer interested readers to Chapter 11
of the survey book~\citep{aggarwal08privacy} for a nice overview of association rule hiding methods.

Oliveira and Zaiane \citep{Oliveira_04p} consider the application of a rotation, additive noise, and multiplicative noise, separately to each
original data {\em attribute}.  As such, their transformation is not guaranteed to preserve Euclidean distance between data {\em tuples}.  
However, Oliveira and Zaiane argue, through experiments, that their overall data perturbation technique preserves the accuracy of
two clustering algorithms.   

Similar to \citep{Oliveira_04p}, Ting {\em et al.} \citep{TFT:2008} considered perturbation of the data attributes
by an orthogonal transformation.  More precisely, Ting {\em et al.} considered left-multiplication of the original data matrix by
a randomly generated orthogonal matrix.  However, they assume the original
data tuples are {\em rows} rather than columns, as we do.  As a result,
Euclidean distance between original data tuples is not preserved, but, sample mean
and covariance is.  If the original data arose as independent samples from
multi-variate Gaussian distribution, then the perturbed data allows inferences
to be drawn about this underlying distribution just as well as the original data.
For all but small or very high-dimensional datasets, their approach is more
resistant to prior knowledge attacks than Euclidean distance-preserving
perturbations.  Their perturbation matrix is $m \times m$ ($m$ is the number of
original data tuples), much bigger than Euclidean distance-preserving
perturbation matrices, $n \times n$ ($n$ is the number of data dimensions).

\noindent \textbf{A survey:} Fung {\em et al.} \citep{FWCY:2010} provided a detailed survey of work related to this paper 
(using the descriptive term "privacy-preserving data publishing").  They discussed a wide range of data perturbation and transformation
techniques, as well as, approaches to breach privacy.  They also discussed scenarios other than the census model, {\em e.g.} multiple release data
publishing and statistical database querying.   

\subsection{Euclidean Distance-Preserving Data Perturbation}

In this part, we describe research most related to this paper.  The majority
of the work focuses on Euclidean distance-preserving data perturbation.

Chen and Liu \citep{Chen_05} observe that some classifiers are invariant
with respect to Euclidean distance between the training tuples.  The authors quantify 
the privacy offered by a Euclidean distance preserving perturbation in terms of the empirical 
covariance matrix with respect to the difference between the original and perturbed data attributes.
The authors' privacy quantification does not take into account prior knowledge, hence, the attack
based on prior knowledge presented in our paper applies directly to the Euclidean distance preserving data
perturbation method of Chen and Liu.  An important issue not discussed by Chen and Liu is how the classifier
learned from perturbed data will be used to classify new tuples.  Perturbing the new tuples and applying the
classifier would produce the same result as if a classifier built from the unperturbed training data was applied to 
the unperturbed new tuples.  But, the process of perturbing the new tuples and applying the classifier need be done
with great care to not leak information that could be used to recover the original training tuples.

Oliveira and Zaiane \citep{Oliveira:03} observe that some clustering algorithms are invariant
with respect to Euclidean distance between data tuples.  The authors quantify privacy using an approach
related to that in Chen and Liu.  Like Chen and Liu, Oliveira and Zaiane do not consider prior knowledge, hence, the attack
based on it presented in our paper applies directly to the Euclidean distance preserving data
perturbation method of Oliveira and Zaiane.

Liu \textit{et al.}~\citep{Liu_06b} developed two types of attacks to
breach the privacy of distance-preserving data perturbation.

\begin{enumerate}
\item Liu developed the \emph{known-sample attack} which assumes that the attacker
has a moderate-sized collection of independent samples chosen i.i.d. from the
same distribution as the private data. By mapping the principal components
of the perturbed data to the principle components of the original data
(estimated from the sample), the attacker can reconstruct the perturbation
matrix and consequently recover the private data. The prior knowledge assumption
made by this attack is different than the assumption in our manuscript
of a very small set of known original tuples.  For example, the known sample prior
knowledge of Liu requires the original dataset and known sample be drawn i.i.d.
(from the same distribution), while the assumption in our manuscript requires no
i.i.d. or any other distribution assumptions.  If, in our manuscript,
we make the additional assumption that the original data is drawn i.i.d., then the
known sample attack of Liu can, in theory, be applied.  But, the attack's accuracy will be
very low as the attack requires a much larger sample than the size of the known tuples we are considering.

\item Liu developed the \emph{known input-output attack} which assumes
that the attacker knows a very small subset of the original (private) data tuples {\em and}
their correspondences to perturbed tuples (i.e. for each known original
tuple, the attacker is assumed to know which is its corresponding perturbed tuple).
Their attack technique is the same as a part of our attack -- choose an orthogonal matrix
randomly from the set of those that satisfy the input-output constraints. Then use a
closed-form expression for the breach probability for each private tuple to choose
the best one to re-estimate.  However, we significantly weaken and make more realistic
(providing an explicit scenario) the prior
knowledge assumption.  We assume only that the attacker knows a very small subset of original
(private) data tuples, but does not know their correspondences to perturbed tuples.  We
extend the attack algorithm of Liu to first infer the correspondences between the known
original tuples and the perturbed tuples.  Also, we provide a complete and rigorous mathematical
analysis of the attack (Liu did not do this).  We also correct a mistake in
$\rho(x_j,\epsilon),$ the probability closed-form expression. Finally, we provide experimental results
(run-time and accuracy) for the attack (Liu did not do this).
\end{enumerate}

Chen {\em et al.} \citep{Chen_07t} also discussed a known input
attack technique.  Unlike ours, they considered a combination of distance-preserving 
data perturbation followed by additive noise.
They also assumed a stronger form of known input prior knowledge: the
attacker knows a subset of private data records {\em and} knows to which
perturbed tuples they correspond.  Finally, they assume that the
number of linearly independent known input data records is no smaller than
the number of data dimensions.  They pointed out that linear
regression can be used to re-estimate private data tuples.

Mukherjee {\em et al.} \citep{MCG:2006} considered the use of
discrete Fourier transformation (DFT) and discrete cosine
transformation (DCT) to perturb the data. Only the high energy
DFT/DCT coefficients are used, and the transformed data in the new
domain approximately preserves Euclidean distance. The DFT/DCT
coefficients were further permuted to enhance the privacy
protection level. Note that DFT and DCT are (complex) orthogonal
transforms. Hence, their perturbation technique can be expressed as
left multiplication by a (complex) orthogonal matrix
(corresponding to the DFT/DCT followed by a perturbation of the
resulting coefficients), then a left multiplication by an identity
matrix with some zeros on the diagonal (corresponding to dropping
all but the high-energy coefficients).  They did not consider
attacks based on prior knowledge. For future work, it would be
interesting to do so.

Turgay {\em et al.} \citep{TPSSL:2008} extended some of the results in \citep{Liu_06b}.  They assume that the
similarity matrix of the original data is made public rather than, $Y$, the
perturbed data itself.  They describe how an attacker, given at least
$n+1$ linearly independent original data tuples {\em and} their corresponding
entries in the similarity matrix, can recover the private data ($n$ is the number
of data dimensions).  Like
Chen {\em et al.}, this differs
from our known input attack in two main ways: (i) we do not require prior
knowledge beyond the known input tuples; (ii) our attack analysis
smoothly encompasses
the case where the number of linearly independent known input tuples is greater
than $n$ as well as less.

Wong {\em et al.} \citep{WCKM:2009} considered data perturbation as a solution to
privacy problems introduced by data outsourcing  wherein an un-trusted party holds the perturbed data and computes
k-nearest-neighbor queries against it on behalf of other parties.  Among other things,
they examined the vulnerabilities of the perturbed data against an attacker armed with
known input prior knowledge (their "level 2" prior knowledge).  Independently of us, they
briefly discussed a basic idea for linking the known inputs to their perturbed
counterparts that is similar to our linking technique (although they provide only a cursory description omitting
many details).\footnote{We described our linking technique in an earlier, unpublished, technical report version of
this paper (citation omitted because of the double-blind nature of this submission).  This report appeared 3 months
after Wong's paper, and, at the time we were unaware of Wong's work.}  They point out how a distance-preserving data
perturbation can be undone if the number of linearly independent known inputs that can be linked to
perturbed tuples exceeds the number of data dimensions.  Their work differs from ours
in that it says nothing about the case where the number of linearly independent, linked
known tuples is less than the number of data dimensions.

Kaplan {\em et al.} \citep{KPSS:2010} considered the estimation of private trajectories (vectors
of real numbers) given various kinds of prior knowledge like Euclidean distances from the private 
trajectories to a known one.  They develop an innovative algorithm that can incorporate a wide variety of types of prior 
knowledge and produce estimates.  The primary differences between Kaplan's and our work
are as follows.  Our work applies to the more general problem where the attacker has only a collection of known 
inputs and does not know their perturbed counterparts.  We develop a novel technique for linking the known
inputs to their perturbed counterparts.  Once this is done, Kaplan's algorithm can be applied to estimate
unknown private tuples from the known input-output pairs. However, unlike Kaplan's, our approach provides precise estimation 
error guarantees, namely, the precise value of the estimation error probability.  Thus, with our approach, the attacker can 
know (in probability) how good each of the estimates is, and, for example, pick the best one.  Through experiments, we found
our approach to be significantly more accurate than Kaplan's.  On the other hand, Kaplan's approach has the advantage over ours 
of being more general in the sense that it can incorporate a larger variety of prior knowledge into the attack.  Our approach is 
tailored to known input-output prior knowledge.   

Before we briefly describe another two attacks
based on independent component analysis (ICA) \citep{Hyv:00}, it
is necessary to give
a brief ICA overview.

\subsubsection{ICA Overview}
Given an $n'$-variate random vector $\mathcal{V}$, one common ICA
model posits that this random vector was generated by a linear
combination of independent random variables, {\em i.e.},
$\mathcal{V}$ $=$ $A\mathcal{S}$ with $\mathcal{S}$ an $n$-variate
random vector with independent components. Typically,
$\mathcal{S}$ is further assumed to satisfy the following
additional assumptions: (i) at most one component is distributed
as a Gaussian; (ii) $n' \geq n$; and (iii) $A$ has rank $n$ (full rank).

One common scenario in practice: there is a set of unobserved
samples (the columns of $n \times q$ matrix $S$) that arose from
$\mathcal{S}$ which satisfies (i) - (iii) and whose components are
independent.  But observed is $n' \times q$ matrix $V$ whose
columns arose as linear combination of the rows of $S$. The
columns of $V$ can be thought of as samples that arose from a
random vector $\mathcal{V}$ which satisfies the above generative
model. There are ICA algorithms whose goal is to recover $S$ and
$A$ from $V$ up to a row permutation and constant multiple.  This ambiguity
is inevitable due to the fact that for any diagonal matrix (with
all non-zeros on the diagonal) $D$, and permutation matrix $P$, if
$A, S$ is a solution, then so is $(ADP)$,
$(P^{-1}D^{-1}\mathcal{S})$.

\subsubsection{ICA Based Attacks}
Liu {\em et al.} \citep{Liu_06} considered matrix multiplicative
data perturbation, $Y = MX$, where $M$ is an $n' \times n$ matrix with each
entry generated independently from the same distribution with mean
zero and variance $\sigma^2$.  They discussed the application of
the above ICA approach to estimate $X$ directly from $Y$:
$\mathcal{S} = \mathcal{X}$, $\mathcal{V}$ $=$ $\mathcal{Y}$, $S =
X$, $V = Y$, and $A = M$. They argued the approach to be
problematic because the ICA generative model imposes assumptions
not likely to hold in many practical situations: the components of
$\mathcal{X}$ are independent with at most one such being Gaussian
distributed.  Moreover, they pointed out that the row permutation
and constant multiple ambiguity further hampers accurate recovery
of $X$. A similar observation is made later by Chen {\em et al.}
\citep{Chen_07t}.

Guo and Wu \citep{Guo_07p} considered matrix multiplicative
perturbation assuming only that $M$ is an $n \times n$ matrix
(orthogonal or otherwise). They assumed the attacker has known
input prior knowledge, {\em i.e.} she knows, $\widetilde{X}$, a
collection of original data columns from $X$. They develop an
ICA-based attack technique for estimating the remaining columns in
$X$. To avoid the ICA problems described in the previous
paragraph, they instead applied ICA {\em separately} to
$\widetilde{X}$ and $Y$ producing representations
$(A_{\widetilde{X}},S_{\widetilde{X}})$ and $(A_Y,S_Y)$.  They
argued that these representations are related in a natural way
allowing $X$ to be estimated.  Their approach, however, will be quite
inaccurate for extremely small numbers of known inputs.  Moreover,
their approach does not provide the attacker with any sort of error
information and she will thus not know which (if any) of her original
data tuple estimates are accurate.

\section{Euclidean Distance-Preserving Perturbation and Privacy Breaches} \label{background}

This section provides: some common notation used throughout the article, the definition of $T$ a
Euclidean distance-preserving data perturbation, the definition of a privacy breach, and
a small example illustrating a Euclidean distance-preserving perturbation.

\subsection{Notation and Conventions}

In the rest of this paper, unless otherwise stated, the following
notations and conventions are used. ``Euclidean distance-preserving'' 
and ``distance-preserving'' are used interchangeably.
All matrices and vectors discussed are assumed to have real
entries (unless otherwise stated).  All vectors are assumed to be
column vectors and $M'$
denotes the transpose of any matrix $M$.  Given a vector $x$, $||x||$
denotes its Euclidean norm.  An $m \times n$ matrix
$M$ is said to be {\em orthogonal} if $M'M = I_n$, the $n \times
n$ identity matrix.\footnote{If $M$ is square, it is orthogonal if
and only if $M' = M^{-1}$ \citep[pg. 17]{Artin_91}.}  The set of
all $n \times n$, orthogonal matrices is denoted by
$\mathbb{O}_n$.

Given $n \times p$ and $n \times q$ matrices $A$ and $B$, let
$[A|B]$ denote the $n \times (p+q)$ matrix whose first $p$ columns
are $A$ and last $q$ are $B$. Likewise, given $p \times n$ and $q
\times n$ matrices $A$ and $B$, let $\left[\begin{array}{c} A \\ B
\end{array}\right]$ denote the $(p+q) \times n$ matrix whose first
$p$ rows are $A$ and last $q$ are $B$.

The data owner's private dataset is represented as an $n \times m$
matrix $X$, with each column a record and each row an attribute
(each record is assumed to be non-zero).
The data owner applies a Euclidean distance-preserving perturbation
to $X$ to produce an $n \times m$ data matrix
$Y$, which is then released to the public or another party for
analysis.  That $Y$ was produced from $X$ by a Euclidean distance-preserving 
data perturbation (but not which one) is also make public.

\subsection{Euclidean Distance-Preserving Perturbation}

A function $H:\Re^n \rightarrow \Re^n$ is Euclidean distance-preserving 
if for all $x,y \in \Re^n$, $||x-y||$ $=$
$||H(x)-H(y)||$. Here
$H$ is also called a {\em rigid motion}.  It has been shown that
any distance-preserving function is equivalent to an orthogonal
transformation followed by a translation \citep[pg. 128]{Artin_91}.
In other words, $H$ may be specified by a pair $(M,v)$ $\in$
$\mathbb{O}_n \times \Re^n$, in that, for all
$x \in \Re^n,$ $H(x)$ $=$ $Mx + v$.
If $v = 0$, $H$ preserve Euclidean length: $||x||$
$=$ $||H(x)||$, as such, it moves $x$ along the surface of the
hyper-sphere with radius $||x||$ and centered at the origin.

Recall that columns of $X$ (denoted $x_1$, $\ldots$, $x_m$) refer to
private data records. And, columns of $T(X)$ $=$ $Y$
(denoted $y_1$, $\ldots$, $y_m$) refer to perturbed data records.
The correspondence between the private and perturbed data records
is not assumed known, {\em e.g.} the perturbed version of $x_i$
is not necessarily $y_i$.  Instead, the columns of $X$ are
transformed using a Euclidean distance-preserving function, then are permuted
to produce the columns of the perturbed dataset $Y$.  Formally, the perturbed
dataset $Y$, is produced as follows.  The private data owner
chooses $(M_T, v_T)$, a secret
Euclidean distance-preserving function, and $\pi$, a secret permutation
of $\{1, \ldots, m\}$.  Then, for $1 \leq i \leq m$, the data owner produces
$y_{\pi(i)}$ $=$ $M_Tx_i + v_T$.

Euclidean distance between the private data tuples is preserved in
the perturbed dataset:
for all $1 \leq i,j \leq m$, $||x_i-x_j||$ $=$ $||y_{\pi(i)}-y_{\pi(j)}||$.
Moreover, if $v_T = 0$, then length of the private data tuples is also
preserved: for all $1 \leq i \leq m$, $||x_i||$ $=$ $||y_{\pi(i)}||$.

\subsection{Privacy Breach} \label{privacybreach}
For simplicity, we assume the attacker produces an estimate for a single unknown original data
tuple.\footnote{As described in Section \ref{introduction}, this can easily be extended to produce estimates for as many unknown original
data tuples as desired.}  Formally, the
attacker will employ a stochastic procedure and produce $1 \leq j \leq m$ and non-zero, $\hat{x} \in \Re^n$.
Here, $\hat{x}$ is an estimate of $x_{\hat{j}}$ (with $\hat{j}$ denoting
$\pi^{-1}(j)$), the private original data tuple that was perturbed to produce
$y_j$.\footnote{The attacker does not
need to know $\hat{j}$; she is merely producing an estimate of the
private data tuple that was perturbed to produce $y_j$.}
Given $\epsilon > 0$, we define a privacy breach as follows.

\begin{definition}
An {\em $\epsilon$-privacy breach} occurs if
$||\hat{x} - x_{\hat{j}}||$ $\leq$ $||x_{\hat{j}}||\epsilon$, {\em i.e.}
if the attacker's estimate is wrong with Euclidean relative
error no more than $\epsilon$.
\end{definition}

In the next section, we describe and analyze the known input attack.
The main focus of analysis concerns, $\rho(\epsilon)$, the probability that an $\epsilon$-privacy
breach occurred.

\subsection{Example}\label{example1}

Figure \ref{figure:example1} illustrates a small private dataset (left) and the result of applying a simple Euclidean
distance-preserving perturbation (a 90-degree clockwise rotation and identity $\pi$).  In general, Euclidean
distance-preserving perturbations can be much more complex than the one illustrated here.

\begin{figure*} [ht!]
\begin{center}
{\includegraphics[scale = 0.25]{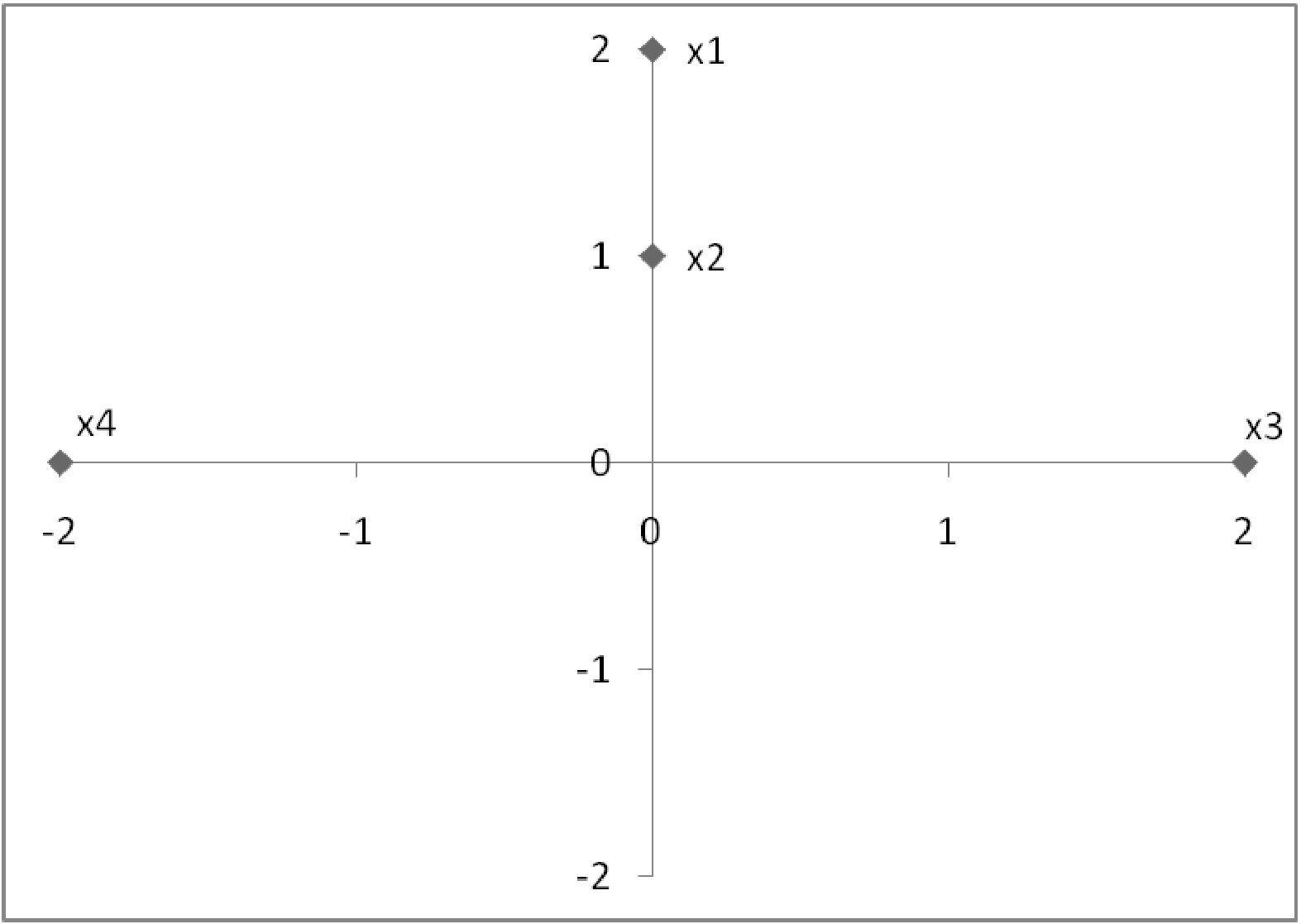}}
{\includegraphics[scale = 0.25]{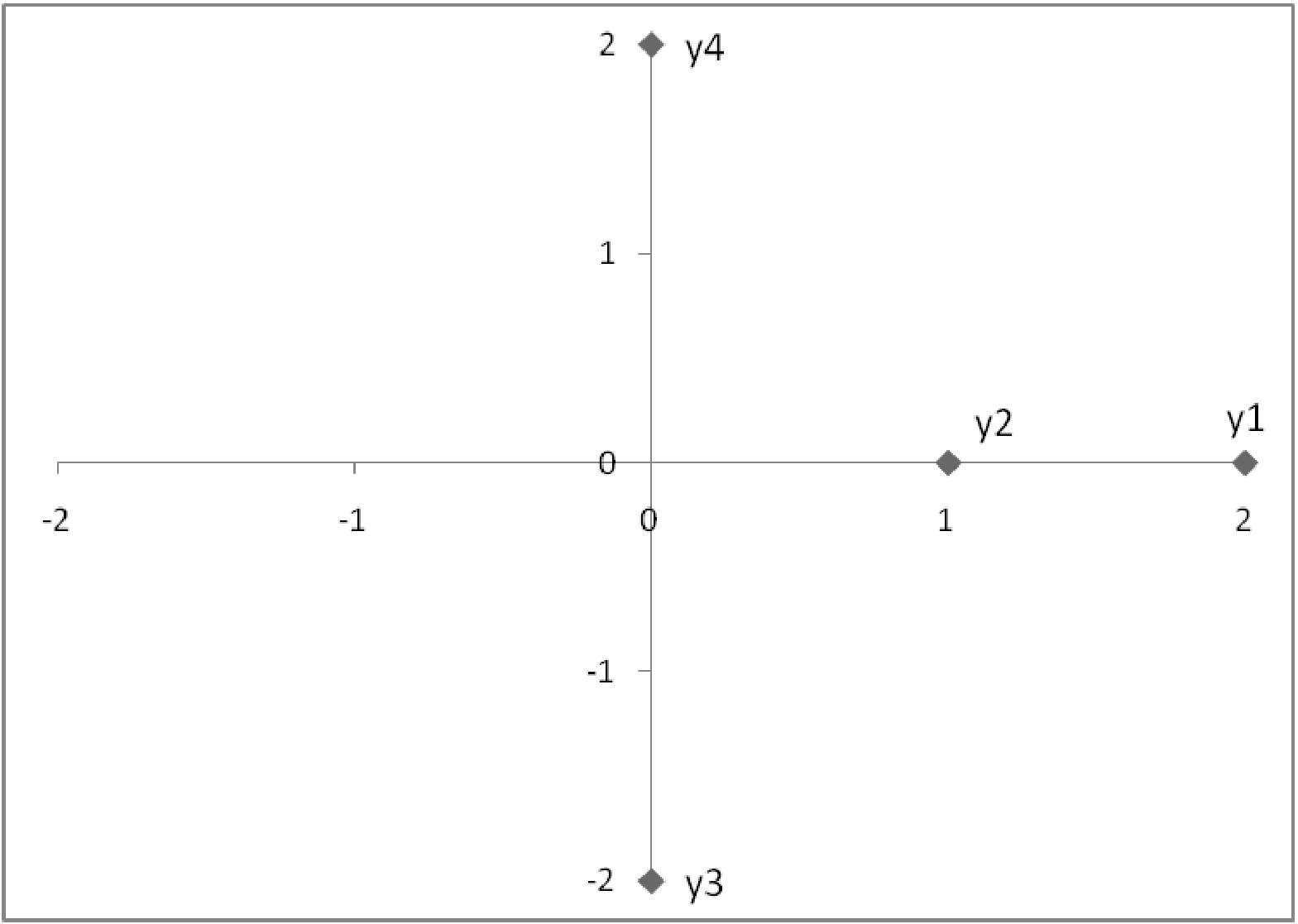}}
\caption{A four record, original private dataset (left) and the result of applying a simple Euclidean distance-preserving 
perturbation: 90-degree clockwise rotation (the perturbed records permutation, $\pi$, is the identity).}\label{figure:example1}
\end{center}
\end{figure*}

\section{Known Input Attack}\label{one_to_one}

For $1 \leq a \leq m-1$, let $X_{a}$ denote the first $a$ columns
of $X$.  The attacker is assumed to know $X_a$ and her attack proceeds in
three steps.  For the remainder of the paper we use interchangeably ``known inputs" and ``known original data tuples".

\begin{enumerate}
\item Infer as many as possible of the input-output mappings in $\pi_a$
(the restriction of $\pi$ to $\{1, \ldots, a\}$), that is, find as many as possible
perturbed counterparts of $X_a$ in $Y$.
\item For each perturbed tuple $y_j$ in $Y$ which is not mapped onto
by $\pi_a$, compute the probability that the following stochastic procedure will result in an $\epsilon$-privacy
breach when estimating the original tuple associated with $y_j$ (the probability calculation is done using a closed-form
expression derived later).
\begin{enumerate}
\item Estimate $M_T$ by choosing a matrix, $\hat{M}$, uniformly from the space of all orthogonal matrices that map the tuples in $X_a$
to their $\pi_a$ counterparts in $Y$ (as computed in step 1).
\item Estimate the original tuple associated with $y_j$ as $\hat{M}'y_j$.
\end{enumerate}
\item Choose the $y_j$
with the highest probability from step 2 and produce $\hat{x}$ $=$ $\hat{M}'y_j$.
\end{enumerate}

The bulk of our work involves the development and analysis of an attack
technique in the case where the data perturbation is assumed to be
orthogonal (does not involve a fixed translation, $v_T = 0$).  The
majority of this section is dedicated to developing and analyzing an
attack in this case.  Then, in Subsection \ref{sec:knownIOgeneral}, we briefly
describe how the attack and analysis can be extended to arbitrary
Euclidean distance-preserving perturbation ($v_T \neq 0$).

\subsection{Inferring $\pi_a$}
\label{sec:fixed}

The attacker may not have enough information to infer $\pi_a$, so, her goal is to infer
$\pi_I$ (the restriction of $\pi$ to $I \subseteq \{1,\ldots,a\}$), for as large an $I$
as possible.  Next, we describe how this goal can be precisely stated as an algorithmic
problem that the attacker can address given her available information.

Given $I$ $\subseteq$ $\{1, \ldots, a\}$, an {\em assignment on} $I$ is
a 1-1 function $\beta:I$ $\rightarrow$ $\{1,\ldots,m\}$.  An assignment $\beta$ on $I$ is
{\em valid} if it satisfies both of the following conditions for all
$i,j \in I$,  (1) $||x_i||$ $=$ $||y_{\beta(i)}||$ and (2)
$||x_i-x_j||$ $=$ $||y_{\beta(i)}-y_{\beta(j)}||$.  Importantly, if $\beta$ is not valid,
it cannot be a correct linkage between tuples in $X_a$ and $Y$, {\em i.e.} $\beta$ $\neq$ $\pi_I$.
As such, there is at least one valid assignment on $I$, namely $\pi_I$, but, there may be more.
If $\beta$ is the only valid assignment on $I$, then it must equal $\pi_I$.

For notational convenience, we say that $I$ is {\em uniquely valid} if there is only one
valid assignment on $I$.  The attacker's goal is to find a {\em maximal} uniquely
valid $I$, {\em i.e.} a uniquely valid $I$ such that there does not
exist uniquely valid $J$ with $|J| > |I|$.
It can be shown that there exists only one maximal uniquely valid subset
of $\{1,\ldots,a\}$.  Thus, the attacker's goal is to find the maximal uniquely
valid subset of $\{1,\ldots,a\}$ along with its corresponding assignment.

The following straight-forward algorithm will meet the attacker's goal by employing a top-down,
level-wise search of the subset space of $\{1, \ldots, a\}$.  The inner for-loop uses an implicit
linear ordering to enumerate the size $\ell$ subsets without repeats and requiring $O(1)$ space.

\begin{algorithm}[h]
\caption{{\footnotesize Overall Algorithm For Finding the Maximal Uniquely Valid Subset}} \label{algorithm:all}{\footnotesize
\begin{algorithmic}[1]

\STATE For $\ell = a, \ldots, 1$, do

\STATE \hspace*{0.5cm} For all $I \subseteq \{1,\ldots,a\}$ and $|I| = \ell$, do

\STATE \hspace*{1cm} If $I$ is uniquely valid, then output $I$ along with its corresponding assignment and terminate the algorithm.

\STATE Otherwise output $\emptyset$.

\end{algorithmic}}
\end{algorithm}

{\em Example revisited -- part 1:} consider the dataset and its perturbed version illustrated in Figure
\ref{figure:example1} and assume that $X_3 = [x_1 x_2 x_3]$ are the known original data tuples
($a=3$).  Algorithm \ref{algorithm:all} proceeds as follows.

\begin{itemize}
\item Check if $I = \{1,2,3\}$ is
uniquely valid.  Since the distances of $y_3$ to $y_2$ and $y_1$ are the same as those of $y_4$
to $y_2$ and $y_1$, then the assignment $\beta: 1 \mapsto 1, 2 \mapsto 2, 3 \mapsto 4$ is valid.
The identity assignment on $I$ is also valid because, in this example, $\pi$ is the identity
permutation.  Thus, $I$ has more than one valid assignment
(is not uniquely valid).
\item  Check if $I = \{1,2\}$ is uniquely valid.  To see that $I$ is uniquely valid note that any
valid assignment, $\beta$, must assign $2$ to itself or else $||x_2|| \neq ||y_{\beta(2)}||$.
And, it can be checked that $\beta(1) \neq 3,4$ in order to satisfy $1=||x_1-x_2|| = ||y_{\beta(1)} - y_{\beta(2)}||$.
Therefore, $\beta$ must be the identity assignment on $I$.
\item The algorithm terminates and outputs $I = \{1,2\}$ as the maximal uniquely valid subset of $\{1,2,3\}$ with
assignment $1 \mapsto 1, 2 \mapsto 2$.  Note: any ordering on the subsets of size two may be considered.  We chose
lexicographic order for simplicity.
\end{itemize}
\qed

Now we develop an algorithm that, given $I \subseteq \{1,\ldots,a\}$,
determines if $I$ is uniquely valid, and, if so, also computes the
corresponding assignment.  The idea is to search the space of all assignments
on $I$ for valid ones.  Once more than one valid assignment is identified,
the search is cut-off and the algorithm outputs that $I$ is not uniquely valid.
Otherwise, exactly one valid assignment, $\pi_I$, will be found.  In this case, the
algorithm outputs that $I$ is uniquely valid and returns the corresponding
assignment.  The algorithm performs a depth-first search with each node, $\mathcal{N}_1$, in the search tree
representing $\emptyset \subseteq I_1 \subseteq I$ and $\beta_1$ a valid assignment on $I_1$.
The search proceeds by considering all $\hat{I_1}$ $=$ $I_1$ $\cup$ $\{\hat{i_1}\}$ where
$\hat{i}$ $\in$ $(I \setminus I_1)$ and all possible ways of extending $\beta_1$
to be a valid assignment, $\hat{\beta_1}$ on $\hat{I_1}$.  In turn, $\hat{I_1}$ and $\hat{\beta_1}$
represent a node, $\hat{\mathcal{N}_1}$, in the search tree immediately below $\mathcal{N}_1$.  If
$\hat{I_1}$ = $I$, then a {\em NumValidAssignFound} counter is incremented.  If the counter exceeds one,
then $I$ has more than one valid assignment, and the search is terminated.

To make this search efficient, we employ a simple, but effective, pruning rule to quickly
eliminate possible extensions of $\beta_1$ that are not valid assignments on $\hat{I_1}$
$=$ $I_1$ $\cup$ $\{\hat{i_1}\}$.  Let $\mathcal{C}(I_1,\beta_1,\hat{i})$ denote the set of all
extensions of $\beta_1$, $\hat{i} \mapsto j$, which appropriately preserve Euclidean distances;
formally put, all $j \in (\{1,\ldots,m\} \setminus \beta_1(I_1))$ which
satisfies both of the following conditions: (1) $||x_{\hat{i}}|| = ||y_{j}||$,
and (2) for all $i_1 \in I_1$, $||x_{i_1} - x_{\hat{i}}||$ $=$
$||y_{\beta_1(i_1)} - y_{j}||.$  It can be shown that $j$ $\notin$ $\mathcal{C}(I_1,\beta_1,\hat{i})$ does
not represent a valid assignment.  Therefore, to enumerate all possible, valid, extensions of
$\beta_1$ on $\hat{I_1}$, it suffices to consider those assignments $\hat{\beta_1}$ on $\hat{I_1}$ which
are of the following form: (i) for all $\ell$ $\in$ $I_1$, $\hat{\beta_1}(\ell) = \beta_1(\ell)$ and
(ii) $\hat{\beta_1}(\hat{i_1})$ $=$ $j$ for some $j$ in $\mathcal{C}(I_1,\beta_1,\hat{i})$.

Algorithms \ref{algorithm:validity_main} and \ref{algorithm:validity_recursive} describe the precise
details of the determination whether $I$ is uniquely valid (namely, the details of the search
discussed in the previous two paragraphs).

\begin{algorithm}[h]
\caption{{\footnotesize Determining Unique Validity Main}} \label{algorithm:validity_main}{\footnotesize
\begin{algorithmic}[1]

\REQUIRE $I \subseteq \{1,\ldots,a\}$.

\STATE Set global variable $NumValidAssignFound$ $=$ $0$.

\STATE Call Algorithm \ref{algorithm:validity_recursive} on
inputs $\emptyset$ and $\beta_{\emptyset}$ ($\beta_{\emptyset}$ denotes the unique
valid assignment on $\emptyset$).

\STATE If $NumValidAssignFound$ $>1$, then return ``$I$ IS NOT UNIQUELY
VALID''.  Else, return ``$I$ IS UNIQUELY VALID WITH ASSIGNMENT'' $\beta_I$.

\end{algorithmic}}
\end{algorithm}

\begin{algorithm}[h]
\caption{{\footnotesize Determining Unique Validity Recursive}} \label{algorithm:validity_recursive}{\footnotesize
\begin{algorithmic}[1]

\REQUIRE $I_1 \subseteq I$ and $\beta_1$ a valid assignment on $I_1$.

\STATE If $I_1 = I$, then

\STATE \hspace*{0.5cm} $NumValidAssignFound = NumValidAssignFound + 1$

\STATE \hspace*{0.5cm} If $NumValidAssignFound == 1$, then set $\beta_I$ to $\beta_1$.

\STATE \hspace*{0.5cm} End If.

\STATE Else, do

\STATE \hspace*{0.5cm} For $\hat{i} \in (I \setminus I_1)$ and as long as
$NumValidAssignFound \leq 1$, do

\STATE \hspace*{1cm} For $j \in \mathcal{C}(I_1,\beta_1,\hat{i})$ and as long as $NumValidAssignFound \leq 1$, do

\STATE \hspace*{1.5cm} Extend $\beta_1$ to $\hat{\beta_1}$ s.t. $\hat{\beta_1}(\hat{i})$ $=$ $j$. Let $\hat{I_1} = I_1 \cup \hat{i}$.

\STATE \hspace*{1.5cm} Call algorithm \ref{algorithm:validity_recursive} on
inputs $\hat{I_1}$ and $\hat{\beta_1}$.

\end{algorithmic}}
\end{algorithm}

{\em Comment:} The order by which the elements of $(I \setminus I_1)$ and
$\mathcal{C}(I_1,\beta_1,\hat{i})$ are chosen in iterating through the for loops in
Algorithm \ref{algorithm:validity_recursive} does not affect the correctness
of the algorithm.  However,
it may affect efficiency.  For simplicity, the loops order the elements
in these sets from smallest to largest index number.

Algorithm \ref{algorithm:all} has worst-case computational complexity $O(m^a)$.
While this is no better than a simple brute-force approach, in our experiments,
quite reasonable running times are observed because few original data tuples
will have the same length and/or few pairs of original data tuples will have
the same Euclidean distance.

\subsection{Known Input-Output Attack}

Assume, without loss of generality, that the attacker applies
Algorithm \ref{algorithm:all} and learns
$\pi_{q}$ ($0 \leq q \leq a$), {\em i.e.} $\{1, \ldots,q\}$
is the maximal uniquely valid subset of $\{1, \ldots, a\}$.  Further, to
simplify notation, we may also assume that $\pi_q(i) = i$.\footnote{This can be achieved by the
attacker appropriately reordering the columns of $X_a$ and $Y$.}  Let $Y_q$ denote the
first $q$ columns of $Y$.  As such, the attacker is
assumed to know $X_q$ and the fact that $Y_q = M_TX_q$ where $M_T$ is an
unknown orthogonal matrix.  Based on this, she will apply an attack,
called the {\em known input-output attack}, to produce $q < j \leq m$, and
$\hat{x}$, which is an estimate of $x_{\hat{j}}$, the private tuple that was perturbed to
produce $y_j$.  The known input-output attack was described in the last two steps in the algorithm at
the beginning of Section \ref{one_to_one}.  More formally, the known input-output attack is as follows.
Let $\mathbb{M}(X_q,Y_q)$ denote the set of all $M \in \mathbb{O}_n$ such that $MX_q = Y_q$.

\begin{enumerate}
\item For each $q < j \leq m$, compute the probability that the following stochastic procedure will result in an
$\epsilon$-privacy breech when estimating $x_{\hat{j}}$.
\begin{enumerate}
\item Estimate $M_T$ by choosing a matrix, $\hat{M}$, uniformly from $\mathbb{M}(X_q,Y_q)$.
\item Estimate $x_{\hat{j}}$ as $\hat{M}'y_j$.\footnote{This is equivalent to a maximum likelihood estimate of
$x_{\hat{j}}$.}
\end{enumerate}
\item Choose the $y_j$
with the highest probability from step 2 and produce $\hat{x}$ $=$ $\hat{M}'y_j$.
\end{enumerate}

A key component of the known input-output attack is the computation of $\rho(x_{\hat{j}},\epsilon)$ $=$
$Pr(||\hat{M}'y_j - x_{\hat{j}}|| \leq ||x_{\hat{j}}||\epsilon)$, the
probability that an $\epsilon$-privacy breach will result from the
attacker estimating $x_{\hat{j}}$ as $\hat{M}'y_j$.  In Section \ref{sec:closed}, we will develop a
closed-form expression for $\rho(x_{\hat{j}},\epsilon)$.  This
expression will only involve information known to the attacker;
therefore, she can choose $q < j \leq m$ so as to maximize
$\rho(x_{\hat{j}},\epsilon)$.  Another key component of the known input-output
algorithm is in choosing $\hat{M}$ uniformly from $\mathbb{M}(X_q,Y_q)$.
In most cases, $\mathbb{M}(X_q,Y_q)$ is
uncountable and it is not obvious how to choose $\hat{M}$.  We will develop and algorithm for
doing so in Section \ref{sec:closed}.  Before getting to Section \ref{sec:closed}, we
discuss some important linear algebra background.

\subsection{Linear Algebra background}

Let $Col(X_q)$ denote the column space of $X_q$ and
$Col_{\bot}(X_q)$ denote its orthogonal complement, {\em i.e.,}
$\{z \in \Re^n:$ $z'w = 0,$ $\forall w \in Col(X_q)\}$.  Likewise,
let $Col(Y_q)$ denote the column space of $Y_q$ and $Col_{\bot}(Y_q)$
denote its orthogonal compliment.  Let $k$ denote the dimension of
$Col(X_q)$.  The ``Fundamental Theorem of
Linear Algebra'' \citep[pg. 95]{S:1986} implies that the dimension
of $Col_{\bot}(X_q)$ is $n-k$.  Since $Y_q = M_TX_q$ and $M_T$ is
orthogonal, then it can be shown that $Col(Y_q)$ has dimension $k$.
Thus, $Col_{\bot}(Y_q)$ has dimension $n-k$.

Let $U_k$ and $V_k$ denote $n \times k$ matrices whose columns form
an orthonormal basis for $Col(X_q)$ and $Col(Y_q)$, respectively.  It
can easily be shown that $Col(M_TU_k)$ $=$ $Col(Y_q)$ $=$ $Col(V_k)$.
Let $U_{n-k}$ and $V_{n-k}$ denote $n \times (n-k)$ matrices whose columns form
an orthonormal basis for $Col_{\bot}(X_q)$ and $Col_{\bot}(Y_q)$,
respectively.  It can easily be shown that $Col(M_TU_{n-k})$ $=$ $Col_{\bot}(Y_q)$
$=$ $Col(V_{n-k})$.

\subsection{A Closed-Form Expression for $\rho(x_{\hat{j}},\epsilon)$}\label{sec:closed}

Now we return to the issue of how to choose $\hat{M}$ uniformly
from $\mathbb{M}(X_q,Y_q)$ and how to compute
$\rho(x_{\hat{j}},\epsilon)$ $=$ $Pr(||\hat{M}'y_j - x_{\hat{j}}||
\leq ||x_{\hat{j}}||\epsilon)$ $=$ $Pr(||\hat{M}'M_Tx_{\hat{j}} - x_{\hat{j}}||)$.

To choose $\hat{M}$ uniformly from $\mathbb{M}(X_q,Y_q)$, the basic idea is to utilize standard
algorithms for choosing a matrix $P$ uniformly from $\mathbb{O}_{n-k}$, the set of all $(n-k)\times(n-k)$ orthogonal
matrices, then apply an appropriately
designed transformation to $P$.  The transformation will be an affine, bijection from
$\mathbb{O}_{n-k}$ to $\mathbb{M}(X_q,Y_q)$.\footnote{That the resulting $\hat{M}$ was chosen {\em uniformly} from
$\mathbb{M}(X_q,Y_q)$ could be more rigorously justified using left-invariance of probability measures and the
Haar probability measure over $\mathbb{O}_{n-k}$.  But, such a discussion is not relevant to this paper and
is omitted.}  The following technical result,
proven in \ref{appendix:surface_area}, provides this transformation.\footnote{We define $\mathbb{O}_{0}$ to contain a single,
empty matrix.  And, for $P \in \mathbb{O}_{0}$, we define $V_{n-k}PU_{n-k}'$ to be the $n \times n$ zero matrix.}

\begin{thm}
\label{thm:IOkey}
Let $L$ be the mapping $P \in \mathbb{O}_{n-k}$ $\mapsto$ $M_TU_kU_k'+V_{n-k}PU_{n-k}'$.  Then, $L$ is an affine
bijection from $\mathbb{O}_{n-k}$ to $\mathbb{M}(X_q,Y_q)$.  And, $L^{-1}$ is the mapping $M \in \mathbb{M}(X_q,Y_q)$
$\mapsto$ $V_{n-k}'MU_{n-k}$.
\end{thm}

\begin{algorithm}
\caption{{\footnotesize Uniform Choice From $\mathbb{M}(X_q,Y_q)$}}
\label{algorithm:UniformChoice}
{\footnotesize
\begin{algorithmic}[1]
\REQUIRE $U_k$, an $n \times k$ matrix whose columns form an orthonormal basis of
$Col(X_q)$, and $M_TU_k$ ($M_T$ is unknown); $U_{n-k}$ and $V_{n-k}$, $n \times (n-k)$ matrices whose columns
form an orthonormal basis of $Col_{\bot}(X_q)$ and $Col_{\bot}(Y_q)$, respectively.

\ENSURE $\hat{M}$ a uniformly chosen matrix from $\mathbb{M}(X_q,Y_q)$.

\STATE Choose $P$ uniformly from $\mathbb{O}_{n-k}$ using
algorithm \citep{H:1978}.

\STATE Set $\hat{M} = L(P)$, {\em i.e.}, $M_TU_kU_k' + V_{n-k}PU_{n-k}'$.

\end{algorithmic}
}
\end{algorithm}

\noindent Two comments are in order regarding Algorithm \ref{algorithm:UniformChoice}.  First, some special cases are 
interesting to highlight: when $k=n$, $\hat{M}$ is chosen as $M_T$; when $k=n-1$, $\hat{M}$ is one of two choices
(one of which equals $M_T$); otherwise, $\hat{M}$ is, in theory, chosen from an uncountable set (containing $M_T$).
Second, it is not obvious how the attacker can compute the inputs to the algorithm,
{\em e.g.} $M_TU_k$.  This issue will be discussed later when spelling out the details of the Known Input Attack
Algorithm, Algorithm \ref{algorithm:IOAttack}.  

Now we develop a closed-form expression for
$\rho(x_{\hat{j}},\epsilon)$.  The key points are outlined, while
a more rigorous justification is provided in
\ref{appendix:surface_area}.  First of all, from Algorithm
\ref{algorithm:UniformChoice},
$\hat{M}$ $=$ $M_TU_kU_k' + V_{n-k}PU_{n-k}'$ where $P$ is chosen uniformly
from $\mathbb{O}_{n-k}$.  Therefore,

\begin{eqnarray*}
\rho(x_{\hat{j}},\epsilon) &=& Pr(||\hat{M}'M_Tx_{\hat{j}} - x_{\hat{j}}|| \leq ||x_{\hat{j}}||\epsilon) \\
&=& Pr(||U_kU_k'x_{\hat{j}} + U_{n-k}P'V_{n-k}'M_Tx_{\hat{j}}- x_{\hat{j}}|| \leq ||x_{\hat{j}}||\epsilon).
\end{eqnarray*}

\noindent Since $\left[\begin{array}{c} U_k' \\ U_{n-k}' \end{array}\right] \in \mathbb{O}_n$, then it can left-multiply
each term in the left $||\ldots||$ of the second probability without changing the equality.  As a result,
the derivation continues

\begin{eqnarray*}
\cdots &=& Pr\left(\left|\left| \left[\begin{array}{c} U_k'x_{\hat{i}} \\ 0 \end{array}\right] +
\left[\begin{array}{c} 0 \\ P'V_{n-k}'M_Tx_{\hat{j}} \end{array}\right] -
\left[\begin{array}{c} U_k'x_{\hat{j}} \\ U_{n-k}'x_{\hat{j}} \end{array}\right]  \right|\right| \leq ||x_{\hat{j}}||\epsilon \right) \\
&=& Pr(||P'V_{n-k}'M_Tx_{\hat{j}} - U_{n-k}'x_{\hat{j}}|| \leq ||x_{\hat{j}}||\epsilon).
\end{eqnarray*}

\noindent Since $Col(M_TU_{n-k})$ $=$ $Col(V_{n-k})$, then there exists $(n-k) \times (n-k)$ matrix $B$ such that
$M_TU_{n-k}B = V_{n-k}$.  It follows that (i) $V_{n-k}'$ $=$ $B'U_{n-k}'M_T'$, (ii) $B$ $=$ $U_{n-k}'M_T'V_{n-k}$.
Thus, $B$ is orthogonal.\footnote{$B'B$ $=$ $B'U_{n-k}'M_T'V_{n-k}$ $=$ $V_{n-k}'V_{n-k}$ $=$ $I_{n-k}$.}  Using (i),
the derivation continues

\begin{eqnarray}
\cdots &=& Pr(||P'B'(U_{n-k}'x_{\hat{j}}) - (U_{n-k}'x_{\hat{j}})|| \leq ||x_{\hat{j}}||\epsilon) \label{IO_derivation1}\\
&=& Pr(||P'(U_{n-k}'x_{\hat{j}}) - (U_{n-k}'x_{\hat{j}})|| \leq ||x_{\hat{j}}||\epsilon) \label{IO_derivation2}
\end{eqnarray}

\noindent where the second equality is due to the fact that $B' \in \mathbb{O}_{n-k}$, and thus $(P'B')$ can be regarded as having been uniformly chosen from
$\mathbb{O}_{n-k}$ just like $P'$ (a rigorous proof of the second equality is
provided in \ref{appendix:surface_area}).  Putting the whole derivation
together,

%%%%%%%%%%%%%%%%%%%%%%
%$z_{\hat{j}}$ denote
%   U_{n-k}'x_{\hat{j}}
%$\epsilon_{\hat{j}}$ denote
%   ||x_{\hat{j}}||\epsilon
%%%%%%%%%%%%%%%%%%%%%

\begin{equation}\label{eq:intuitive_closed}
\rho(x_{\hat{j}},\epsilon) = Pr(P \mbox{ uniformly chosen from } \mathbb{O}_{n-k} \mbox{ satisfies } ||P'(U_{n-k}'x_{\hat{j}}) - (U_{n-k}'x_{\hat{j}})|| \leq ||x_{\hat{j}}||\epsilon).
\end{equation}

Let $S_{n-k}(||U_{n-k}'x_{\hat{j}}||)$ denote the hyper-sphere in
$\Re^{n-k}$ with radius $||U_{n-k}'x_{\hat{j}}||$ and centered at the
origin.  Since $P$
is chosen uniformly from $\mathbb{O}_{n-k}$, then any point on the
surface of $S_{n-k}(||U_{n-k}'x_{\hat{j}}||)$ is equally likely to be
$P'(U_{n-k}'x_{\hat{j}})$. Let $S_{n-k}(U_{n-k}'x_{\hat{j}},||x_{\hat{j}}||\epsilon)$
denote the ``hyper-sphere cap'' consisting of all points in $S_{n-k}(||U_{n-k}'x_{\hat{j}}||)$ with
distance from $U_{n-k}'x_{\hat{j}}$ no greater than
$||x_{\hat{j}}||\epsilon$. Therefore, (\ref{eq:intuitive_closed}) becomes

\begin{eqnarray}\label{eq:intuitive_closed2}
\rho(x_{\hat{j}},\epsilon) &=& Pr(\mbox{a uniformly chosen point
on } S_{n-k}(||U_{n-k}'x_{\hat{j}}||) \mbox{ is also in }
S_{n-k}(U_{n-k}'x_{\hat{j}},||x_{\hat{j}}||\epsilon)) \nonumber \\
&=&\frac{SA(S_{n-k}(U_{n-k}'x_{\hat{j}},||x_{\hat{j}}||\epsilon))}{SA(S_{n-k}(||U_{n-k}'x_{\hat{j}}||))}
\end{eqnarray}

\noindent where $SA(.)$ denotes the surface area of a subset of a
hyper-sphere.\footnote{$S_1(||U_{1}'x_{\hat{j}}||)$ consists of two points.
We define $\frac{SA(S_{1}(U_{1}'x_{\hat{j}},||x_{\hat{j}}||\epsilon))}{SA(S_{1}(||U_{1}'x_{\hat{j}}||))}$
as 0.5 if $S_1(U_{1}'x_{\hat{j}},||x_{\hat{j}}||\epsilon)$ is one point, and as 1 otherwise.  Moreover,
we define $\frac{SA(S_{0}(U_{0}'x_{\hat{j}},||x_{\hat{j}}||\epsilon))}{SA(S_{0}(||U_{0}'x_{\hat{j}}||))}$ as 1.}
Based on equations (\ref{eq:intuitive_closed2}), we prove, in \ref{appendix:surface_area},
the following closed form expression, for
$\rho(x_{\hat{j}},\epsilon)$, where, $\Gamma(.)$ denotes the
standard gamma function,
$ac_{[]-1}(x)$ denotes
$arccos\left(\left[\frac{||x_{\hat{j}}||\epsilon}{||U_{n-k}'x_{\hat{j}}||\sqrt{2}} \right]^2 - 1\right)$,
and $ac_{1-[]}(x)$ denotes
$arccos\left(1-\left[\frac{||x_{\hat{j}}||\epsilon}{||U_{n-k}'x_{\hat{j}}||\sqrt{2}} \right]^2\right)$.

{\footnotesize
\begin{equation}\label{eq:closed1}
\rho(x_{\hat{j}},\epsilon) = \left\{ \begin{array}{ll}
1 & \mbox{if $n-k=0$;} \\
1 & \mbox{if $||x_{\hat{j}}||\epsilon \geq ||U_{n-k}'x_{\hat{j}}||2$ and $n-k \geq 1$;} \\
0.5 & \mbox{if $||x_{\hat{j}}||\epsilon < ||U_{n-k}'x_{\hat{j}}||2$ and $n-k = 1$;} \\
1 - (1/\pi)ac_{[]-1}(x) & \mbox{if $||U_{n-k}'x_{\hat{j}}||\sqrt{2} < ||x_{\hat{j}}||\epsilon < ||U_{n-k}'x_{\hat{j}}||2$ and $n-k = 2$;} \\
1 - \frac{(n-k-1)\Gamma([n-k+2]/2)}{(n-k)\sqrt{\pi}\Gamma([n-k+1]/2)}\int_{\theta_1=0}^{ac_{[]-1}(x)}sin^{n-k-1}(\theta_1)\,d\theta_1  & \mbox{if $||U_{n-k}'x_{\hat{j}}||\sqrt{2} < ||x_{\hat{j}}||\epsilon < ||U_{n-k}'x_{\hat{j}}||2$ and $n-k \geq 3$;} \\
(1/\pi)ac_{1-[]}(x) & \mbox{if $||x_{\hat{j}}||\epsilon \leq ||U_{n-k}'x_{\hat{j}}||\sqrt{2}$ and $n-k = 2$;} \\
\frac{(n-k-1)\Gamma([n-k+2]/2)}{(n-k)\sqrt{\pi}\Gamma([n-k+1]/2)}\int_{\theta_1=0}^{ac_{1-[]}(x)}sin^{n-k-1}(\theta_1)\,d\theta_1 & \mbox{if $||x_{\hat{j}}||\epsilon \leq ||U_{n-k}'x_{\hat{j}}||\sqrt{2}$ and $n-k \geq 3$.} \end{array} \right.
\end{equation}
}

\noindent {\em Comment:} it can be shown that $||U_{n-k}'x_{\hat{j}}||$ is the distance from $x_{\hat{j}}$
to its closest point in $Col(X_q)$ (the column space of $X_q$).  Thus, the sensitivity of a tuple to breach
is dependent upon its length relative to its distance to the column space of $X_q$.  In particular, if the
distance from $x_{\hat{j}}$ to the column space of $X_q$ is sufficiently small, less than $(||x_{\hat{j}}||\epsilon)/2$,
then the breach probability is one from the second case in equation (\ref{eq:closed1}).

Recall that the attacker seeks to use the closed-form expressions
for $\rho(x_{\hat{j}},\epsilon)$ to decide for which $q < j
\leq m$ does $\hat{x} = \hat{M}'y_{j}$ produce the best
estimation of $x_{\hat{j}}$.  This is naturally done by choosing
$j$ to maximize $\rho(x_{\hat{j}},\epsilon)$.  To allow for this,
observe that $||x_{\hat{j}}||\epsilon$ and $||U_{n-k}'x_{\hat{j}}||$
equal\footnote{$M_Tx_{\hat{j}}$ $=$ $y_j$, so,
$||x_{\hat{j}}||$ $=$ $||M_Tx_{\hat{j}}||$ $=$ $||y_j||$.  Moreover,
as shown earlier, there exists $B \in \mathbb{O}_{n-k}$
such that $V_{n-k}'$ $=$ $B'U_{n-k}'M_T'$.  Thus,
$||U_{n-k}'x_{\hat{j}}||$ $=$ $||B'U_{n-k}'M_T'M_Tx_{\hat{j}}||$ $=$
$||V_{n-k}'y_j||$.} $||y_j||\epsilon$ and
$||V_{n-k}'y_j||$, respectively, which are known to the attacker.
Therefore, (\ref{eq:closed1}) can be rewritten as follows,
where $ac_{[]-1}(y)$ denotes
$arccos\left(\left[\frac{||y_j||\epsilon}{||V_{n-k}'y_j||\sqrt{2}} \right]^2 - 1\right)$,
and $ac_{1-[]}(y)$ denotes
$arccos\left(1-\left[\frac{||y_j||\epsilon}{||V_{n-k}'y_j||\sqrt{2}} \right]^2\right)$.

{\footnotesize
\begin{equation}\label{eq:closed2}
\rho(x_{\hat{j}},\epsilon) = \left\{ \begin{array}{ll}
1 & \mbox{if $n-k=0$;} \\
1 & \mbox{if $||y_j||\epsilon \geq ||V_{n-k}'y_j||2$ and $n-k \geq 1$;} \\
0.5 & \mbox{if $||y_j||\epsilon < ||V_{n-k}'y_j||2$ and $n-k = 1$;} \\
1 - (1/\pi)ac_{[]-1}(y) & \mbox{if $||V_{n-k}'y_j||\sqrt{2} < ||y_j||\epsilon < ||V_{n-k}'y_j||2$ and $n-k = 2$;} \\
1 - \frac{(n-k-1)\Gamma([n-k+2]/2)}{(n-k)\sqrt{\pi}\Gamma([n-k+1]/2)}\int_{\theta_1=0}^{ac_{[]-1}(y)}sin^{n-k-1}(\theta_1)\,d\theta_1  & \mbox{if $||V_{n-k}'y_j||\sqrt{2} < ||y_j||\epsilon < ||V_{n-k}'y_j||2$ and $n-k \geq 3$;} \\
(1/\pi)ac_{1-[]}(y) & \mbox{if $||y_j||\epsilon \leq ||V_{n-k}'y_j||\sqrt{2}$ and $n-k = 2$;} \\
\frac{(n-k-1)\Gamma([n-k+2]/2)}{(n-k)\sqrt{\pi}\Gamma([n-k+1]/2)}\int_{\theta_1=0}^{ac_{1-[]}(y)}sin^{n-k-1}(\theta_1)\,d\theta_1 & \mbox{if $||y_j||\epsilon \leq ||V_{n-k}'y_j||\sqrt{2}$ and $n-k \geq 3$.} \end{array} \right.
\end{equation}
}

Now we put together all the parts and provide the pseudo-code of the full known input attack algorithm 
(Algorithm \ref{algorithm:IOAttack}).  Before doing so, first note that
$U_k$, $U_{n-k}$, $V_k$, and $V_{n-k}$ can be computed from $X_q$ and $Y_q$
using standard procedures \citep{S:1986}.  Second, $M_TU_k$ $=$ $Y_qA$
where $A$ is an $q \times k$ matrix that can be computed\footnote{Since $Col(U_k)$ $=$ $Col(X_q)$,
then by solving $k$ systems of linear equations (one for each column of $U_k$), a $q \times k$ matrix $A$
can be computed such that $X_qA = U_k$.} from $U_k$ and $X_q$.  Third, a recursive
procedure for computing (\ref{eq:closed2}) is described in \ref{appendix:surface_area}.

\begin{algorithm}
\caption{{\footnotesize Known Input Attack Algorithm}}
\label{algorithm:IOAttack}
{\footnotesize
\begin{algorithmic}[1]
\REQUIRE $Y$, $\epsilon \geq 0$, and $X_a$.

\ENSURE $a < j \leq m$ and $\hat{x}$ $\in$ $\Re^n$ the
corresponding estimate of $x_{\hat{j}}$.

\STATE Compute $Y_q = M_TX_q$ (where $1 \leq q \leq a$) using Algorithm \ref{algorithm:all}.

\STATE Compute $U_k, V_k, U_{n-k}, V_{n-k},$ and $M_TU_k$ as described earlier.

\STATE For each $q < j \leq m$ do

\STATE \hspace*{0.5cm} Compute $\rho(x_{\hat{j}},\epsilon)$
using (\ref{eq:closed2}) as described in \ref{appendix:surface_area}.

\STATE End For.

\STATE Choose the $j$ from the previous loop producing the largest
$\rho(x_{\hat{j}},\epsilon)$.

\STATE Choose $\hat{M}$ uniformly from $\mathbb{M}(X_q,Y_q)$ by applying
Algorithm \ref{algorithm:UniformChoice}.

\STATE Set $\hat{x}$ $\leftarrow$ $\hat{M}'y_j$.

\end{algorithmic}
}
\end{algorithm}

{\em Comment:} The $\epsilon$-privacy breach probability $\rho(\epsilon)$ equals
$\max_{q < j \leq m}\rho(x_{\hat{j}},\epsilon)$.

{\em Example revisited - part 2:} consider the dataset and its perturbed version illustrated in Figure
\ref{figure:example1} with known original tuples $x_1, x_2, x_3$.   Part 1 of this example showed how
Algorithm \ref{algorithm:all} inferred the following mappings to perturbed tuples $x_1 \mapsto y_1$ and $x_2 \mapsto y_2$, hence
$X_q = [x_1 x_2]$ and $Y_q = [y_1 y_2].$  Consider perturbed tuple $y_4$.  The Known Input Attack will compute $\hat{x}$ an
estimate of the original tuple $x_{\hat{4}}$ associated with $y_4$ ($x_{\hat{4}}$ = $x_4$ in this case) and $\rho(x_{\hat{4}},\epsilon)$, 
the $\epsilon$-privacy breach probability.  Since $x_1$ and $x_2$ are linearly dependent, $k=1$, so, $n-k = 1$ and the second or third cases
of Equation \ref{eq:closed2} apply.  It can be shown that  $V'_{n-k}$ $=$ $V'_{1}$ $=$ $[0, 1]$.  So, 
$||y_4||\epsilon$ $=$ $2\epsilon$ and $||V'_{n-k}y_4||2$ $= 4$.  Therefore, if $\epsilon \geq 2$, the second case applies and
$\rho(x_{\hat{4}},\epsilon) = 1$, else, the third case applies and $\rho(x_{\hat{4}},\epsilon) = 0.5$.

There are only two Euclidean distance preserving transformations fixing the origin that satisfy the input-output constraints 
$x_1 \mapsto y_1$ and $x_2 \mapsto y_2$: the 90-degree clockwise rotation (the actual perturbation applied) and the 90-degree counter-clockwise rotation, 
these are the elements of $\mathbb{M}(X_2,Y_2)$.   
So $\hat{x}$ is chosen randomly between the inverse of these transformations applied to $y_4$ resulting in $\hat{x} = [2,0]'$ or $[-2,0]'$.
If $\epsilon \geq 2$, then either of these choices represent an $\epsilon-$privacy breach so $\rho(x_{\hat{4}},\epsilon)=1$.  If $\epsilon < 2$,
then only one of these choices, $[2,0]'$, represent a breach so $\rho(x_{\hat{4}},\epsilon) =0.5.$
\qed

\subsection{Known Input Attack on General Distance-Preserving Data Perturbation}
\label{sec:knownIOgeneral}

Previously, we considered the case where the data
perturbation is assumed to be orthogonal (does not involve a fixed
translation, $v_T = 0$).  Now we briefly discuss how the attack technique
and its analysis can be extended to arbitrary
Euclidean distance-preserving perturbation ($v_T \neq 0$).

\noindent \textbf{Extending the algorithms for inferring $\pi_a$:}
Since the length of the
private data tuples may not be preserved, then the definition of validity
in Section \ref{sec:fixed} must be changed: $\beta$ on $I$ is
{\em valid} if $\forall i, j \in I$, $||x_i-x_j||$ $=$
$||y_{\beta(i)}-y_{\beta(j)}||$.  As well, the definition of
$\mathcal{C}(I_1,\beta_1,\hat{i})$ (given $I_1 \subseteq I$, $\beta_1$ a valid
assignment on $I_1$, and $\hat{i}$ $\in$ $(I \setminus I_1)$), must change:
the set of all $j \in (\{1,\ldots,m\} \setminus \beta_1(I_1))$ such that
for all $i_1 \in I_1$, $||x_{i_1} - x_{\hat{i}}||$ $=$
$||y_{\beta_1(i_1)} - y_{j}||.$   With these changes, Algorithms \ref{algorithm:all},
\ref{algorithm:validity_main}, and \ref{algorithm:validity_recursive} work correctly as stated.

\noindent \textbf{Extending the known input attack:}
The basic idea is simple and relies on the
fact that the same $v_T$ is added to all tuples in the perturbation of $X_q$.
Fix one tuple, say $x_1$ and $y_1$, and consider the following differences
$x^{-}_{1} = (x_{q} - x_1)$, $\ldots$, $x^{-}_{q-1} =
(x_{q} - x_{q-1})$ and $y^{-}_{1} = (y_{q} - y_1)$, $\ldots$,
$y^{-}_{q-1} = (y_{q} - y_{q-1})$. Let $X^{-}_{q-1}$ denote the
matrix with columns $x^{-}_{1}, \ldots, x^{-}_{q-1}$ and
$Y^{-}_{q-1}$ denote the matrix with columns $y^{-}_{1}, \ldots,
y^{-}_{q-1}$.  Observe that $Y^{-}_{q-1}$ $=$ $M_TX^{-}_{q-1}$, hence,
the attack and its analysis from the orthogonal data perturbation case can be
applied.  The details are straight-forward and are omitted for brevity.
However, a caveat is in order.  The attack depends on the choice of the
tuple to fix. Therefore, the attacker examines them all and chooses the
highest privacy breach probability.

\section{Experiments and Discussion}\label{sec:experiments}

The experiments are designed to assess the computational efficiency of the overall known input attack and its
effectiveness at breaching privacy.  We performed two sets of experiments: (a) those involving only the known
input attack, and (b) those comparing the known input attack with the attack of Kaplan {\em et al.} \citep{KPSS:2010}. 
In both sets of experiments, we used two datasets as the original, private data tuples $X$: 1) a 100,000 tuple synthetic
dataset generated from a 100-variate Gaussian distribution\footnote{The mean vector is specified by independently 
generating 100 numbers from a univariate Gaussian
with mean zero and variance one.  The covariance matrix is specified by (i) independently generating 100 data tuples each 
with 100 independently generated
entries a from a univariate Gaussian with mean zero and variance one, (ii) computing the empirical covariance of this 
100 tuple dataset.}; 2) the Letter Recognition dataset, 20,000 tuples and 16 numeric
attributes, from UCI machine learning repository \citep{FA:2010} -- we removed tuples which were duplicated over 
the numeric attributes yielding a
final dataset of 18,668 tuples. The attacks were implemented in Matlab 7 (R14) and
all experiments were carried out on a Thinkpad laptop with 1.83GHz Intel Core 2 CPU, 1.99GB RAM, and WindowsXP system.  We
did not compare our attack technique against the ICA-based attack in \citep{Guo_07p} and the known sample attack in 
\citep{Liu_06b} because the extremely small size of the known inputs will render these attacks ineffective.
In all Figures, the error bars show one standard deviation above and below the average.

\subsection{Experiments Only Involving the Known Input Attack}

The first experiment fixes $X$ and its perturbed version $Y$, but changes the
number of known input tuples, $a$. It proceeds by carrying out ten trials as follows.
Select $a$ linearly independent tuples randomly from $X$ (these become the know inputs).
Use Algorithm \ref{algorithm:all} to compute $I$, the maximal uniquely valid assignment.
Use steps 2-5 in Algorithm \ref{algorithm:IOAttack} to compute the $\rho(\epsilon)$, the
$\epsilon$-privacy breach probability (a closed-form was given immediately above
Algorithm \ref{algorithm:IOAttack}).

To measure the accuracy of the attack, we report the average of $\rho(\epsilon)$ and $|I|$ over all ten trials.
To measure the efficiency, we report the average time taken to compute $I$ (the rest is
ignored as the overall attack computation time is dominated by Algorithm \ref{algorithm:all}).
In Figures \ref{figure:known_input_gaussian} and \ref{figure:known_input_letter}, results
are shown with $\epsilon = 0.15$.  In Figure \ref{figure:known_input_epsilonchange}, accuracy results are shown with
varying $\epsilon$ and $a$ fixed at four.

\begin{figure*} [ht!]
\begin{center}
{\includegraphics[scale = 0.3]{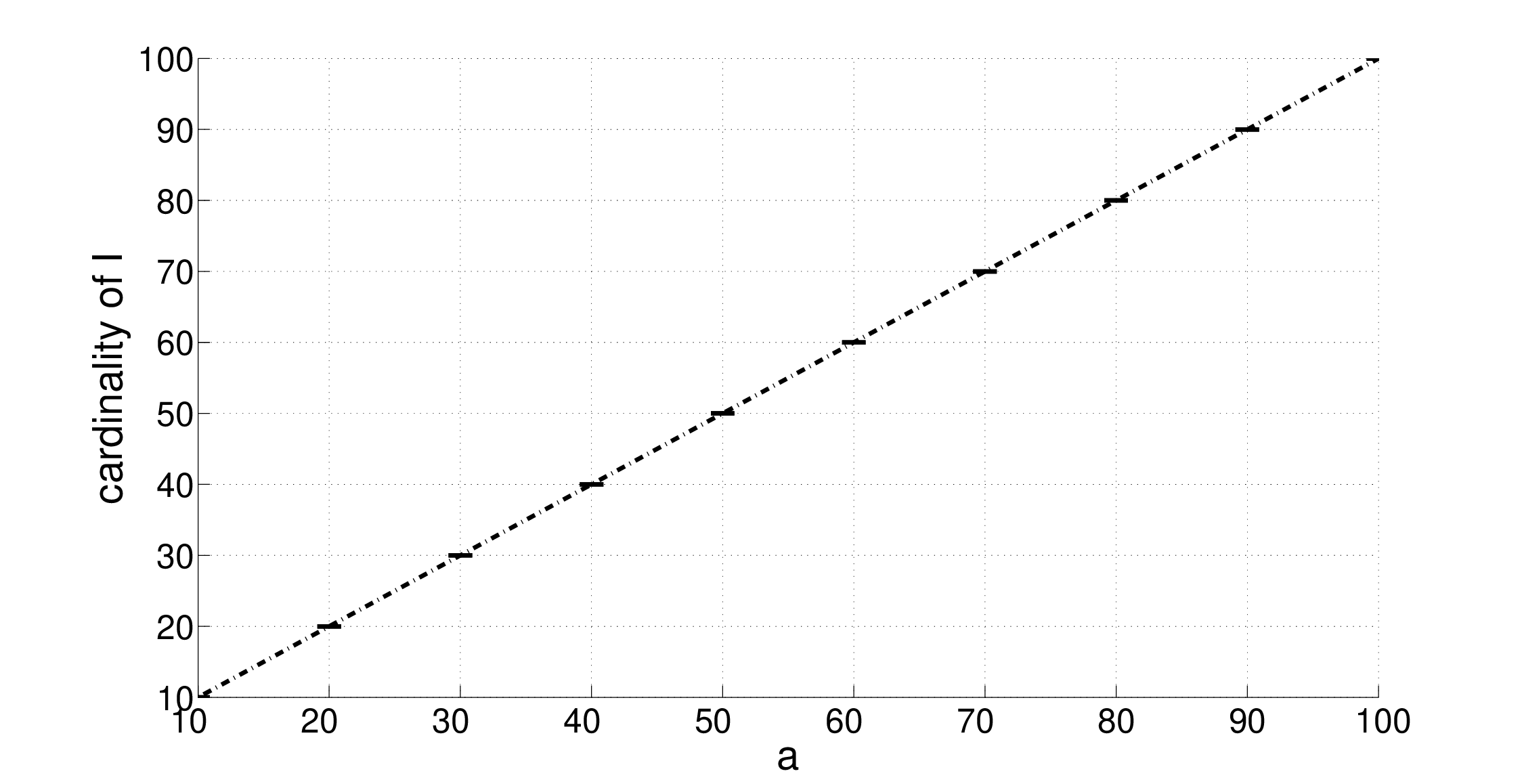}}
{\includegraphics[scale = 0.3]{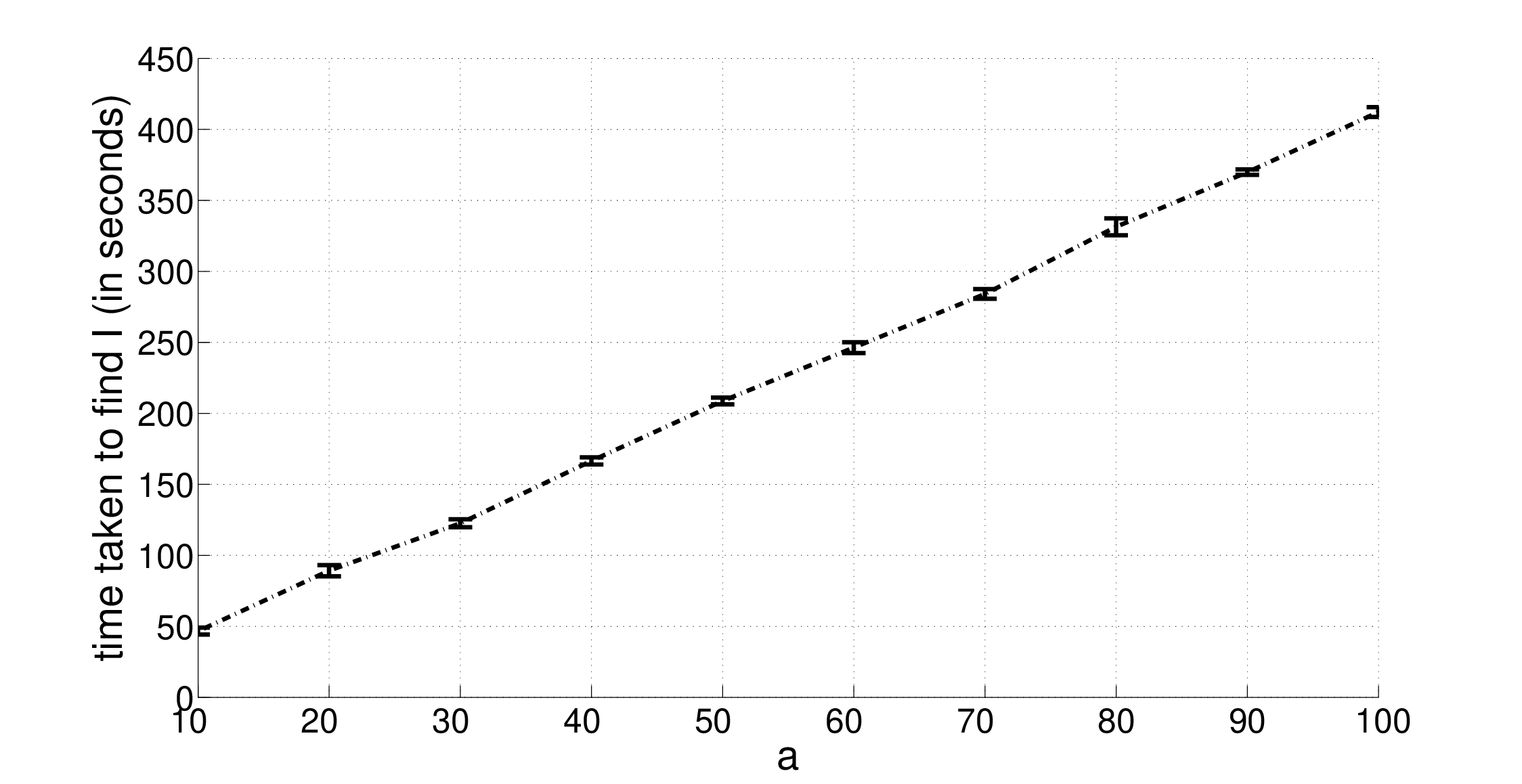}}
{\includegraphics[scale = 0.3]{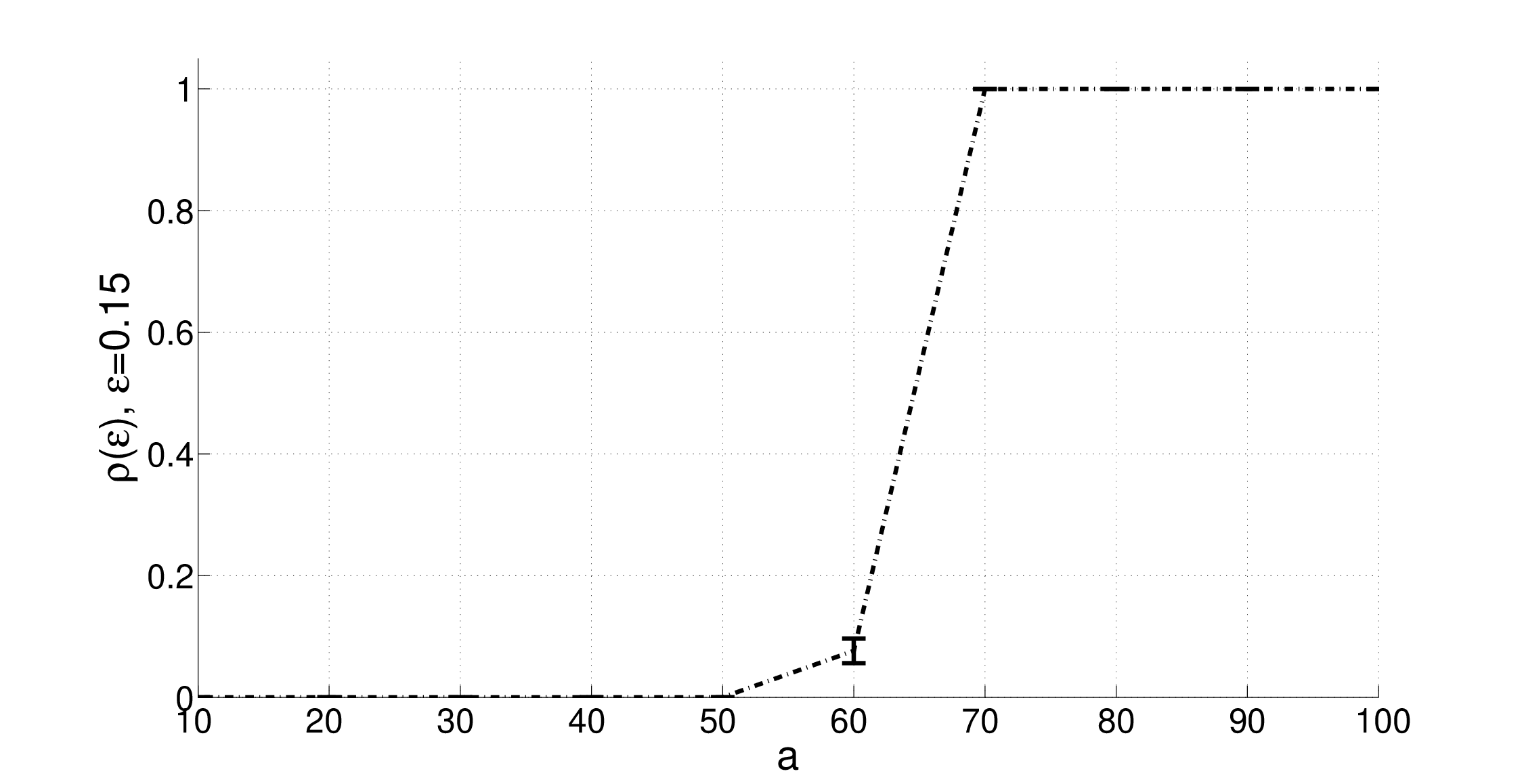}}
\caption{Known input attack on Gaussian data with different number of known inputs and
$\epsilon = 0.15$.}\label{figure:known_input_gaussian}
\end{center}
\end{figure*}

\begin{figure*} [ht!]
\begin{center}
{\includegraphics[scale = 0.3]{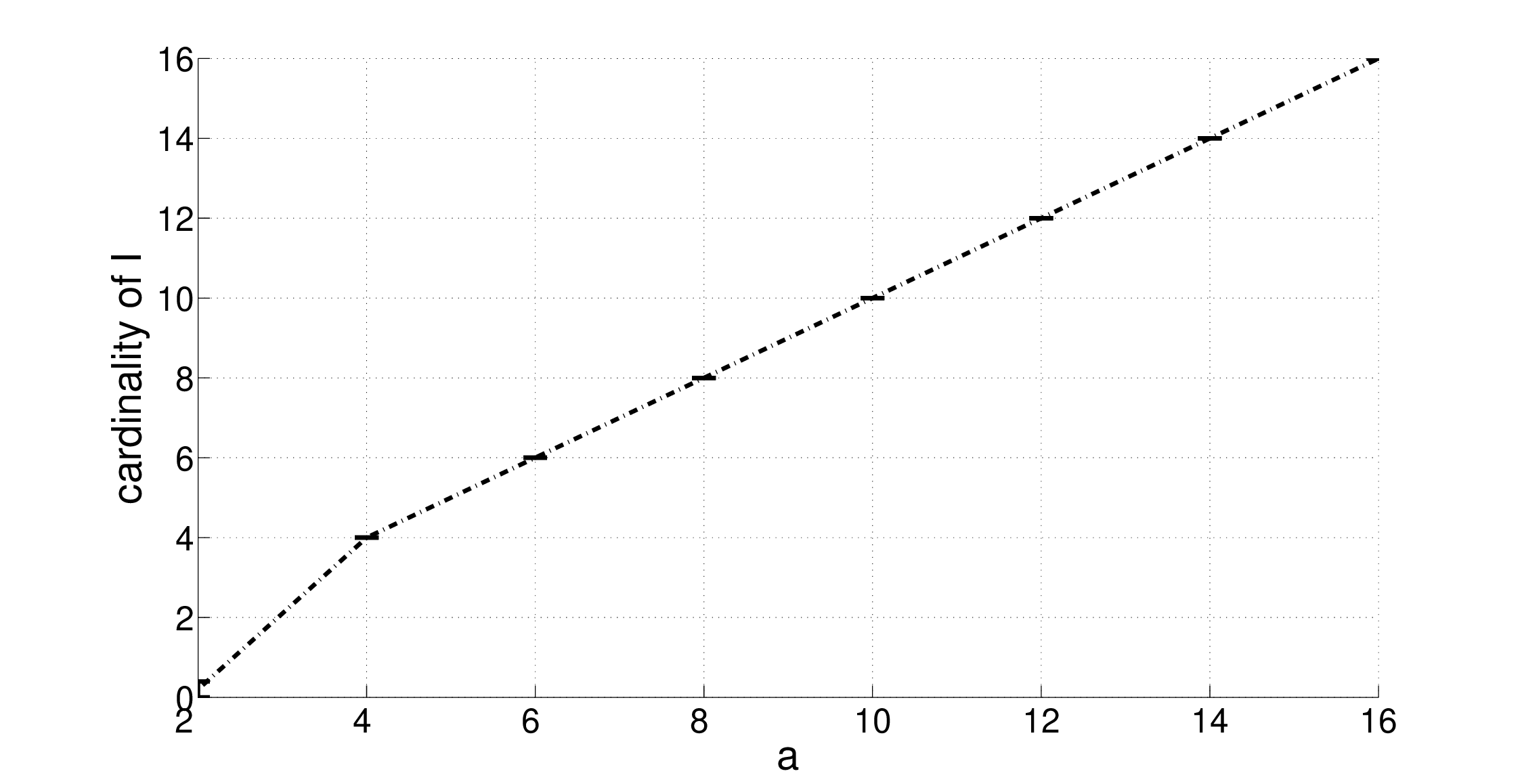}}
{\includegraphics[scale = 0.3]{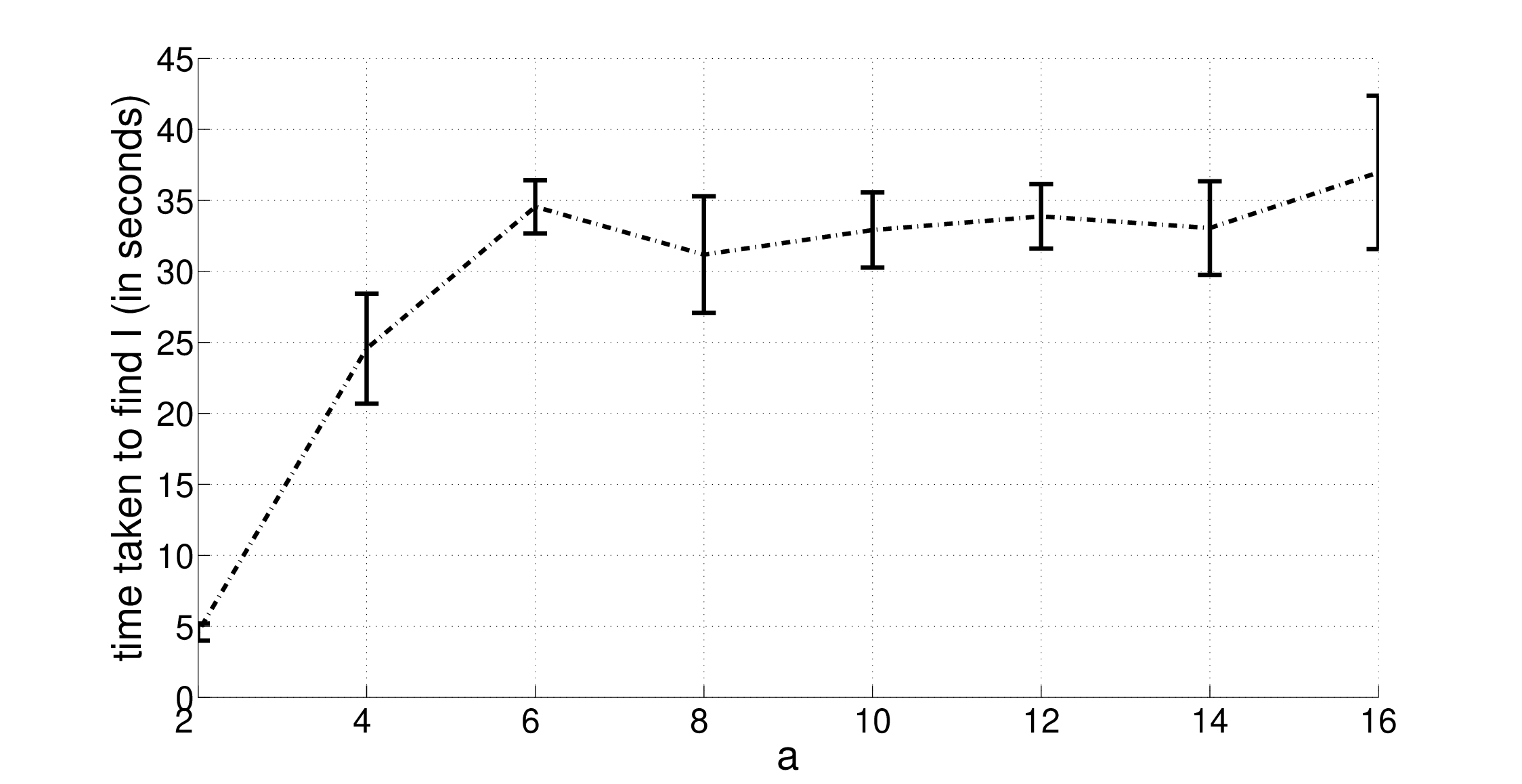}}
{\includegraphics[scale = 0.3]{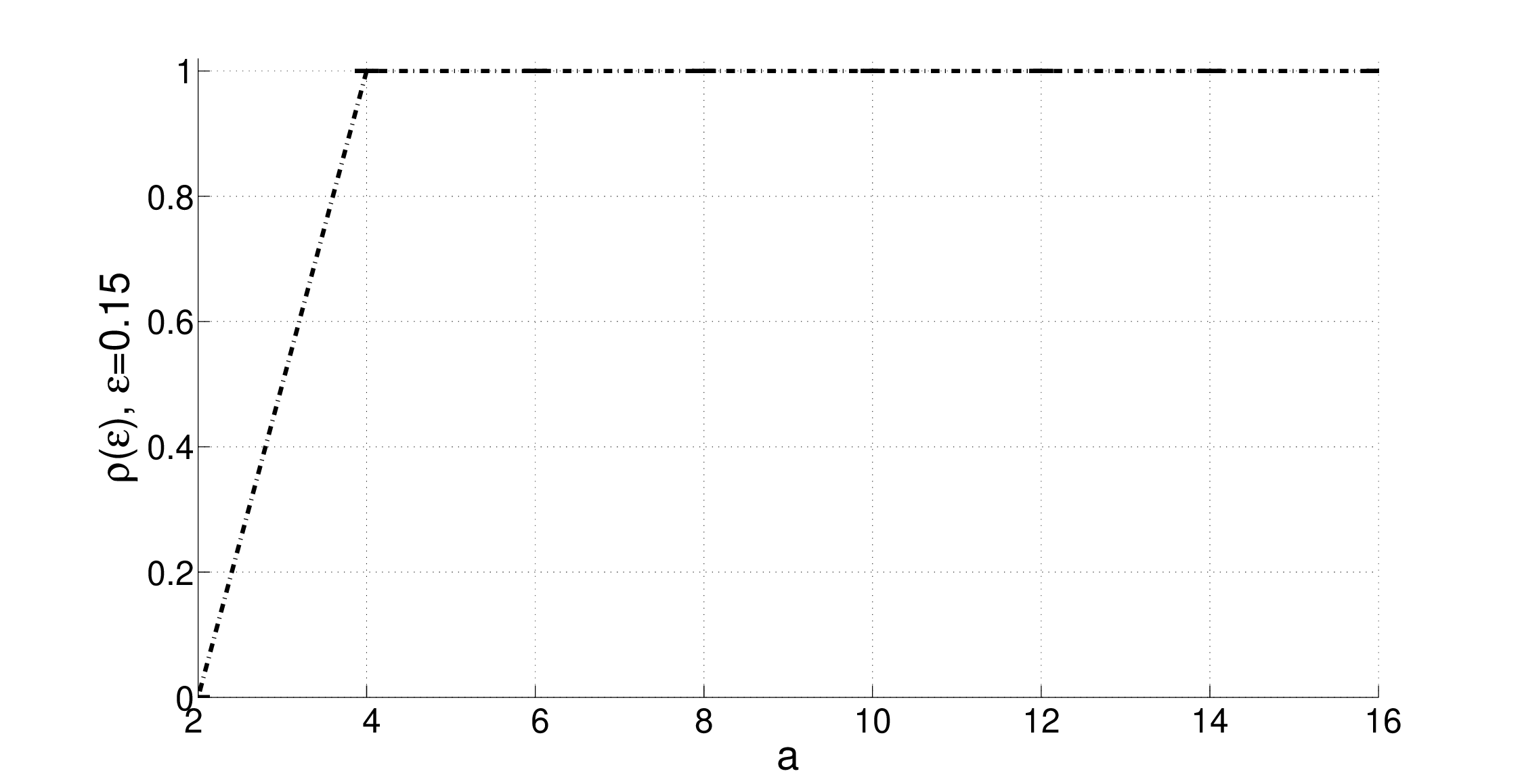}}
\caption{Known input attack on Letter Recognition data with different number of known inputs and $\epsilon = 0.15$.}\label{figure:known_input_letter}
\end{center}
\end{figure*}

%\begin{figure*} [ht!]
%\begin{center}
%{\includegraphics[scale = 0.2]{FIGS/known_input_gaussian_time_datasizechange.eps}}
%{\includegraphics[scale = 0.2]{FIGS/known_input_letter_original_time_datasizechange.eps}}
%\caption{Known input attack on Gaussian (left) and Letter Recognition data (right) with varying size, but fixed number of known inputs. The error bar shows one standard deviation above and below the average value.}\label{figure:known_input_datasizechange}
%\end{center}
%\end{figure*}

%\begin{figure*}[ht!]
%\begin{center}
%\subfigure[]{\includegraphics[scale = 0.16]{FIGS/known_input_letter_original_privacybreach_epsilonchange.eps}\label{figure:known_input_epsilonchange}}
%\subfigure[]{\includegraphics[scale = 0.16]{FIGS/known_input_gaussian_time_datasizechange.eps}
%             \includegraphics[scale = 0.16]{FIGS/known_input_letter_original_time_datasizechange.eps}\label{figure:known_input_datasizechange}}
%\caption{(a) Known input attack on Letter Recognition data with fixed data size, fixed number of known inputs, but varying $\epsilon$. (b) Known input attack on Gaussian (left) and Letter Recognition data (right) with varying size, but fixed number of known inputs. }\label{}
%\end{center}
%\end{figure*}

\begin{figure*}[ht!]
\begin{center}
\includegraphics[scale = 0.3]{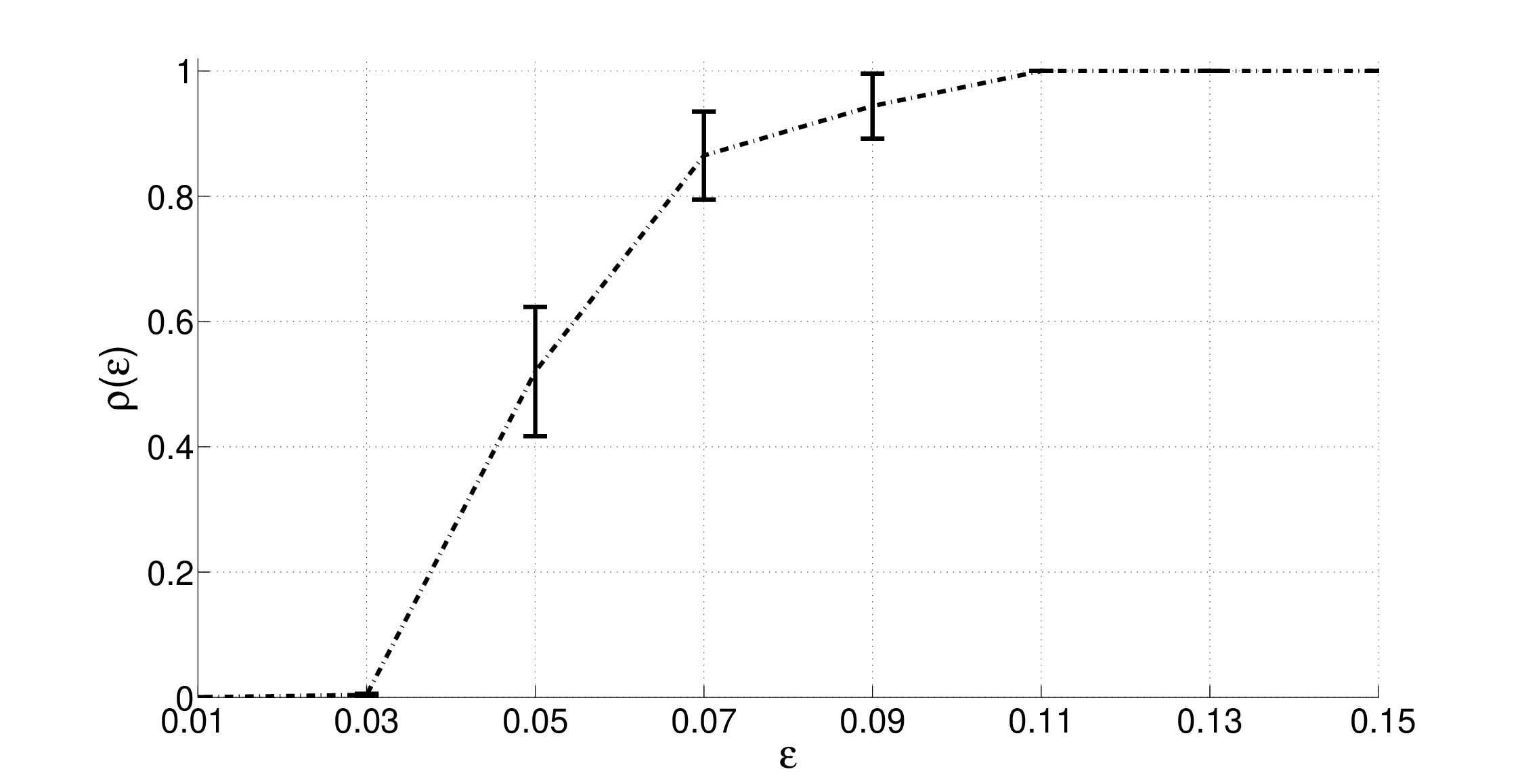}
\caption{Known input attack on Letter Recognition data with fixed data size, fixed number of known inputs $a=4$, but varying $\epsilon$. }
\label{figure:known_input_epsilonchange}
\end{center}
\end{figure*}

\begin{figure*}[ht!]
\begin{center}
\includegraphics[scale = 0.3]{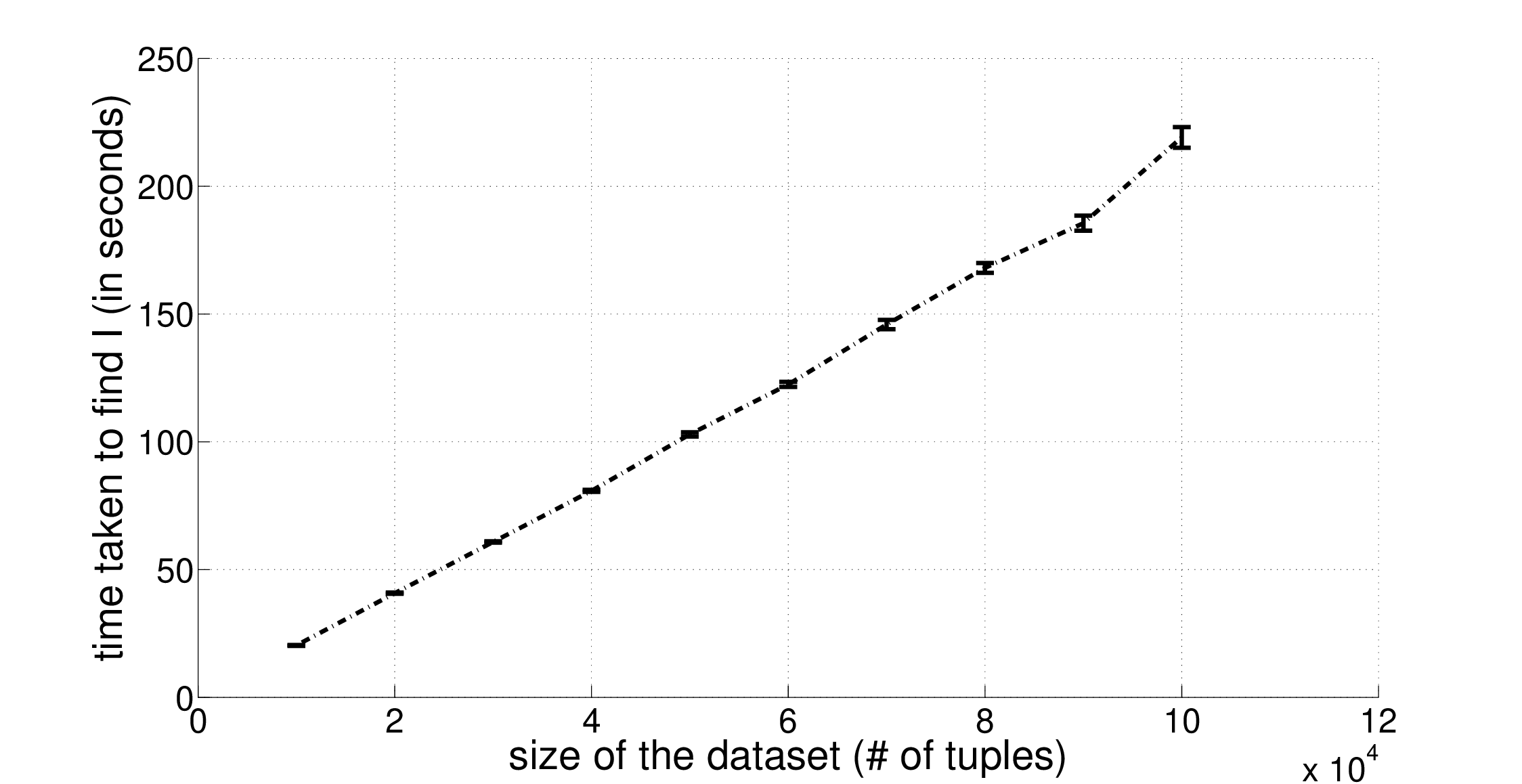}
\includegraphics[scale = 0.3]{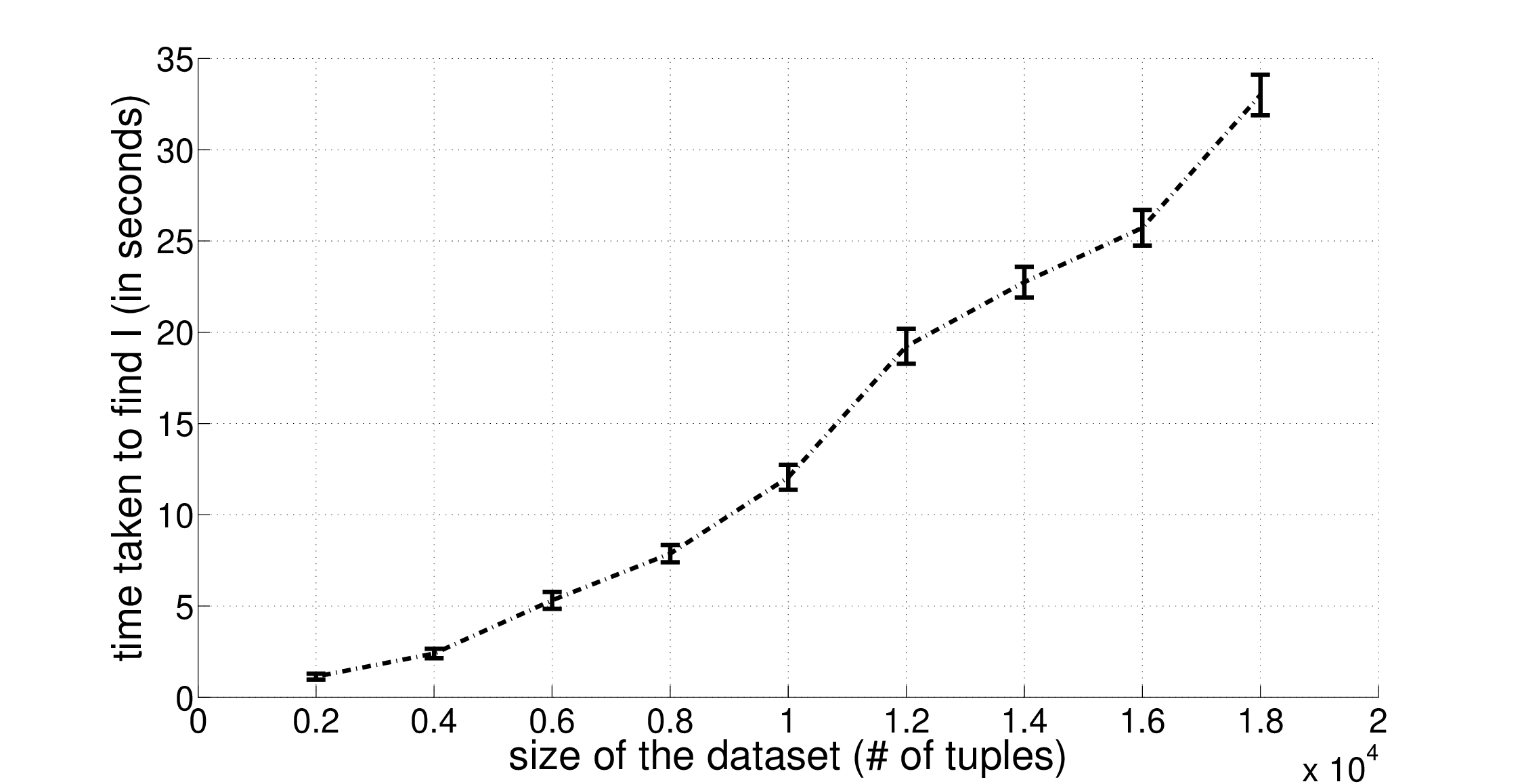}
\caption{Known input attack on Gaussian (left) and Letter Recognition data (right) with varying size, but fixed number of known inputs $a=50,10$ (respectively) and fixed $\epsilon = 0.15$.}
\label{figure:known_input_datasizechange}
\end{center}
\end{figure*}

%\begin{figure*}[ht!]
%\begin{minipage}[t]{0.30\linewidth}
%\centering
%\includegraphics[scale = 0.2]{FIGS/known_input_letter_original_privacybreach_epsilonchange_v2%.eps}
%\caption{Known input attack on Letter Recognition data with fixed data size, fixed number of known inputs $a=4$, but varying %$\epsilon$. }
%\label{figure:known_input_epsilonchange}
%\end{minipage} \hspace{0.1in}
%\begin{minipage}[t]{0.68\linewidth}
%\centering
%\includegraphics[scale = 0.2]{FIGS/known_input_gaussian_time_datasizechange_v2.eps}
%\includegraphics[scale = 0.2]{FIGS/known_input_letter_original_time_datasizechange_v2.eps}
%\caption{Known input attack on Gaussian (left) and Letter Recognition data (right) with varying size, but fixed number of %known inputs $a=50,10$ (respectively) and fixed $\epsilon = 0.15$.}
%\label{figure:known_input_datasizechange}
%\end{minipage}
%\end{figure*}

The second experiment fixes the number of known input tuples (and $\epsilon$ at 0.15) but changes the size of the original
data $X$ in order to
assess the computational efficiency of the attack. For the Gaussian data, it uses the first $k$ tuples as $X$ where $k$
takes a value in $\{10000, 20000, \dots, 100000\}$. Then, the attack proceeds by carrying out the following operations
ten times. Select $a=50$ linearly independent tuples randomly from $X$ and use Algorithm \ref{algorithm:all} to compute
the maximal uniquely valid assignment $I$. The average time taken to compute $I$ is given
in Figure~\ref{figure:known_input_datasizechange} top. For the Letter Recognition data, $k$ takes a value
in $\{2000, 4000, \dots, 18000\}$ and the attack randomly select $a=10$ linearly independent tuples as the known inputs. The
average time taken to find $I$ is given in Figure~\ref{figure:known_input_datasizechange} bottom.

%\noindent {\em Discussion:} The purpose of the experiments was to assess the effectiveness of the known input attack in
%breaching privacy.  To evaluate the results as such we adopt the commonly used viewpoint in the cryptographic community
%that the attacker is granted a wide degree of latitude.  Under this viewpoint, the burden rests with the privacy defenders to
%argue that the attacker cannot breach privacy in any plausible situation.  So, we believe the experimental results should
%be evaluated from the standpoint of whether they support the conclusion that the attack can breach privacy in plausible
%situations.

Regarding the known input attack accuracy, the linking phase of the attack (Algorithm \ref{algorithm:all}), exhibits
excellent performance.  For synthetic data, its performance is perfect in that all known input tuples have their
corresponding perturbed tuple inferred (see Figure~\ref{figure:known_input_gaussian} top).  For real data, its performance is
nearly perfect -- see Figure~\ref{figure:known_input_letter} top.  As expected, $\rho(\epsilon)$
approaches one as $a$ increases see Figures \ref{figure:known_input_gaussian} and \ref{figure:known_input_letter} bottom.
Interestingly on the synthetic dataset, the transition from $\rho(\epsilon) = 0 \rightarrow 1$ occurs very sharply around $a=60$.
Moreover, on the real dataset, $\rho(\epsilon)=1$ with $a$ as small as 4 (and we also observe in
Figure \ref{figure:known_input_epsilonchange} that the probability remains fairly high for $\epsilon$ as small as 0.07).

Regarding computational efficiency, the algorithm appears to require quite reasonable time in all cases
observed, {\em e.g.} less that 450 seconds on the synthetic dataset with 100 known tuples 
(see Figure~\ref{figure:known_input_gaussian} middle)
and less than 45 seconds on the real dataset with 16 known inputs (see Figure~\ref{figure:known_input_letter} middle). With
respect to known input set size ($a$), the average computation time exhibits a linear (synthetic data) or slower
(real data) trend (see Figure~\ref{figure:known_input_gaussian} and Figure~\ref{figure:known_input_letter} middle).
With respect to dataset size (number of private data tuples), the average computation time exhibits a clear linear trend
for both synthetic and real data (see Figure~\ref{figure:known_input_datasizechange}). These results demonstrate that, despite the high worst-case computational complexity, the computation times on both real and
synthetic data are quite reasonable.

The experimental results support the conclusion that the attack can breach privacy in plausible situations.  For example,
on the 16-dimensional, 18688 tuple real dataset, the known input attack achieves a privacy breach with probability one
using four known inputs and less than 30 seconds of run-time.

\subsection{Experiments Comparing the Known Input-Output Attack with Kaplan's Attack}

We compare the accuracy of the attacks with respect to an attacker's goal of producing a single perturbed tuple and an 
estimate of its unperturbed counterpart.  Since Kaplan's attack does not provide the attacker with any 
means to know how good the estimate is, the attacker has no reason to choose one perturbed tuple over another, hence, we
assume the attacker picks randomly.  On the other hand, our attack provides the attacker with breech proababilities, so,
the attacker chooses the perturbed tuple to maximize the breech probability (as done in Algorithm \ref{algorithm:IOAttack}).  
 
The experiments proceed as follows.  $X$ and $Y$ are fixed and $a$, the number of known
input tuples, is varied.  100 trials are carried out as follows.  Select $a$ linearly independent tuples randomly from $X$ 
(these become the know inputs) and do both of the following.  (i) Choose a tuple $y_{\ell}$ randomly from $Y$ whose unperturbed 
counterpart $x_{\hat{\ell}}$ in $X$ is not among the $a$ known inputs.  Use Kaplan's attack\footnote{With a learning rate of 0.05 
and 500 iterations, values observed empirically to produce the best results} to produce an estimate, $\hat{x}$, of $x_{\hat{\ell}}$.  Record
the Euclidean relative error of the estimate, $||x_{\hat{\ell}} - \hat{x}||/||x_{\hat{\ell}}||$.  (ii) Use our Algorithm 
\ref{algorithm:IOAttack}\footnote{With $\epsilon =$ 0.15 and 0.1 for the Gaussian and Letter data, respectively.} 
to choose the tuple $y_j$ from $Y$ with maximum $\epsilon-$privacy breech probability and whose unperturbed counterpart $x_{\hat{j}}$ does not 
appear among the $a$ known inputs, then produce an estimate $\hat{x}$ of $x_{\hat{j}}$.  Record the Euclidean relative error of the 
estimate: $||x_{\hat{j}} - \hat{x}||/||x_{\hat{j}}||$.      
  
%\begin{figure*} [ht!]
%\begin{center}
%{\includegraphics[scale = 0.2]{FIGS/time_spent_known_input_output_emre_kaplan_gaussian_v2.eps}}
%{\includegraphics[scale = 0.2]{FIGS/time_spent_known_input_output_emre_kaplan_letter_v2.eps}}
%\caption{A comparison of Kaplan's attack and our Algorithm \ref{algorithm:IOAttack}) 
%steps 6-7.  The top and bottom charts show the run-times on the Gaussian and Letter recognition datasets, respectively.}\label{figure:attack_comparison_run_times}
%\end{center}
%\end{figure*}

\begin{figure*} [ht!]
\begin{center}
{\includegraphics[scale = 0.3]{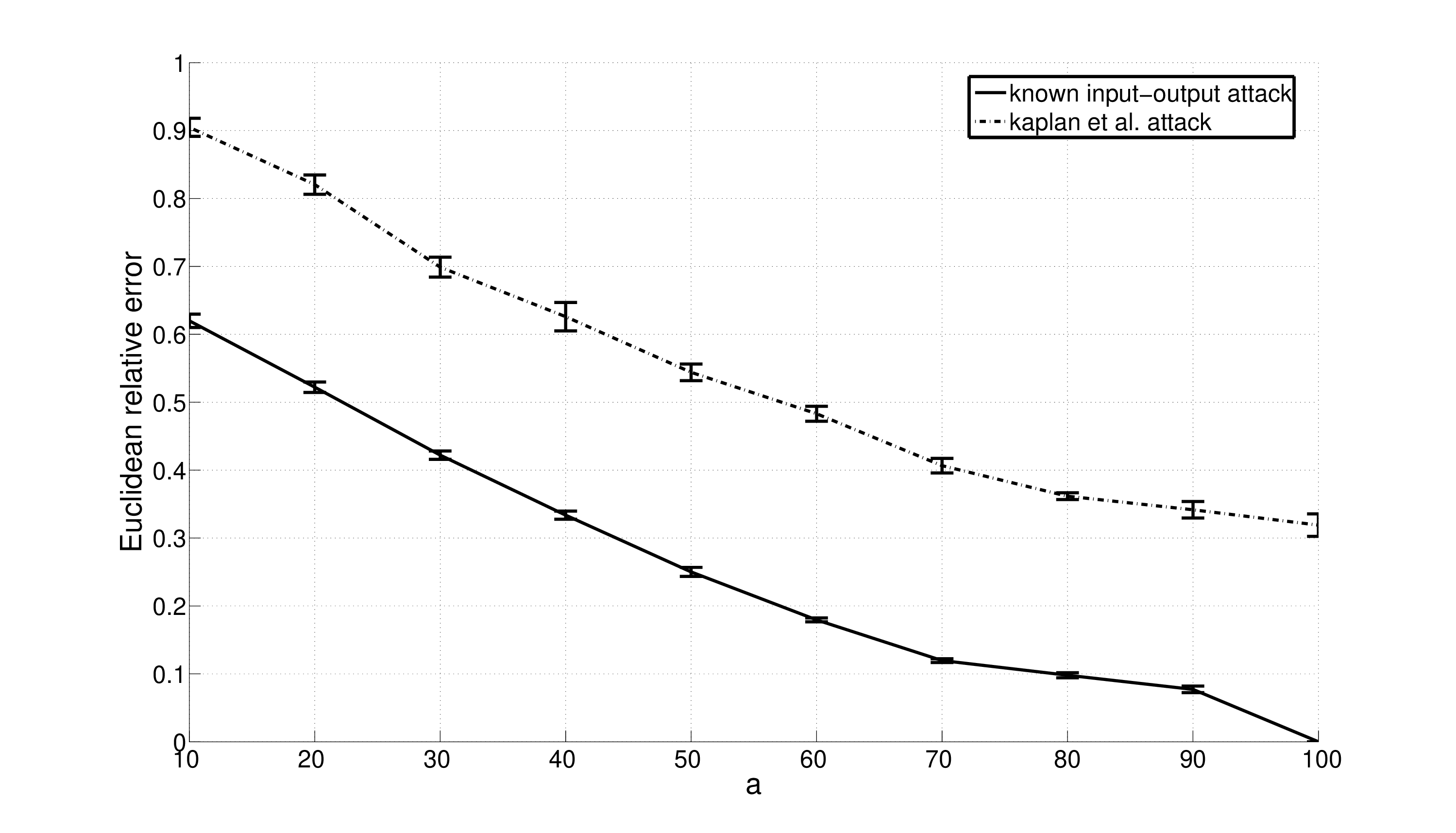}}
{\includegraphics[scale = 0.3]{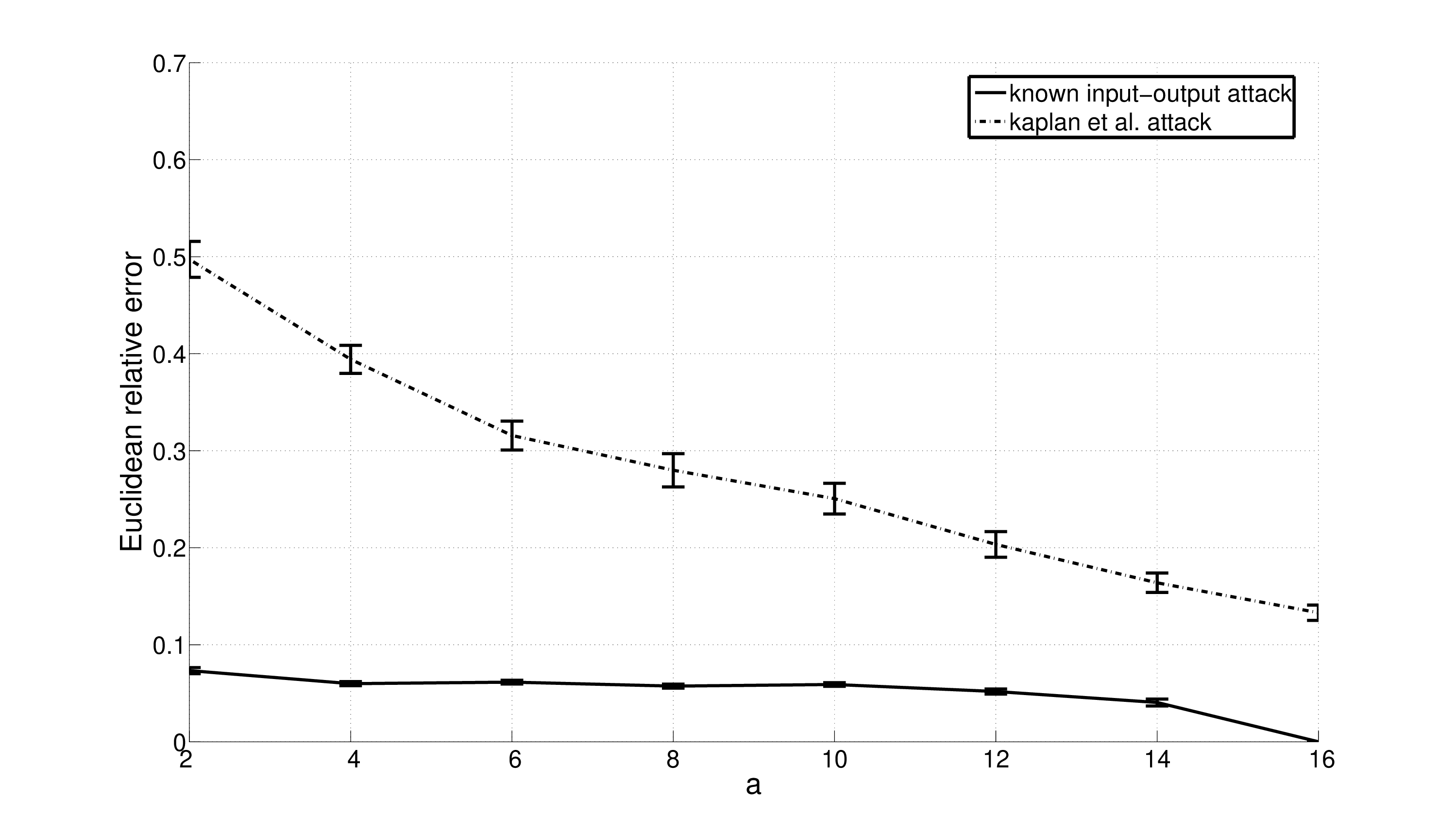}}
\caption{A comparison of Kaplan's attack and our Algorithm \ref{algorithm:IOAttack}) in terms of average
relative error over 100 trials.  The top and bottom charts show the average error on the Gaussian and Letter recognition 
datasets, respectively.}\label{figure:attack_comparison_accuracy}
\end{center}
\end{figure*}

Figure \ref{figure:attack_comparison_accuracy} shows the average relative error of the attacks on both datasets.  From the Figure 
it is clear that our approach allows the attacker to produce a significantly more accurate estimate.  We do not 
provide a figure comparing the run-times of the attacks because Kaplan's computes an estimate from only one perturbed tuple 
while ours, in effect, computes an estimate from all perturbed tuples.  However, the time required by our algorithm to produce an 
estimate from a single, randomly chosen, perturbed tuple is 100 to 1000 times faster 
than Kaplan's. 

\section{Conclusion}\label{sec:discussion}

We examined the vulnerability of Euclidean distance-preserving data perturbation
when a small set of original data tuples are known to the attacker.  We developed a
stochastic technique allowing the attacker to estimate, for each perturbed tuple,
the original unknown data tuple and calculate the probability that the estimation
results in a privacy breach.  For perturbations which fix the origin, this probability is dependent
on the length of the original tuple relative to its distance from the column space of the known inputs.
Therefore, the probability increases as the number of known original tuples does, reaching one when
the number of linearly independent known original tuples reaches the number of data dimensions.  The assumption of fixing
the origin can be dropped, resulting in a slightly more complicated breach probability
calculation.  Our experiments
on real and synthetic data showed that even with the number of known original tuples significantly smaller than the
number of data dimensions, privacy is breached with high probability.  For example,
on a real 16-dimensional dataset, 4 known original tuples is enough for the attacker to 
estimate an unknown original tuple with less than 7\% error with probability exceeding 0.8.

We conclude the paper by pointing to an interesting direction for future work, extending
techniques in this paper to apply to random projection data perturbation: $Y$ $=$
$\ell^{-1/2}\hat{R}X$ where $\hat{R}$ is an $\ell \times n$ matrix
with each entry generated independently and from a standard normal distribution
(this type of data perturbation for $\ell \leq n$ was discussed in \citep{Liu_06}).
It can be shown that
matrix $R$ is orthogonal on expectation and the probability of
orthogonality approaches one exponentially fast with $\ell$. By
increasing $\ell$, the data owner can guarantee that distances are
preserved with arbitrarily high probability.  However, such an increase
intuitively would seem to increase the vulnerability with respect to a
known input attack.   Some preliminary results along these lines can be found in \citep{LGK:2008}.

% use section* for acknowledgment

\bibliographystyle{elsarticle-num}
\bibliography{dpt_journalWOKnownSample}

\appendix

\section{Supplementary Material} \label{appendix:surface_area}

\subsection{Known Input Attack: Proof of Theorem \ref{thm:IOkey}}

\textit{Theorem} \ref{thm:IOkey}:
Let $L$ be the mapping $P \in \mathbb{O}_{n-k}$ $\mapsto$ $M_TU_kU_k'+V_{n-k}PU_{n-k}'$.  Then, $L$ is an affine
bijection from $\mathbb{O}_{n-k}$ to $\mathbb{M}(X_q,Y_q)$.  And, $L^{-1}$ is the mapping $M \in \mathbb{M}(X_q,Y_q)$
$\mapsto$ $V_{n-k}'MU_{n-k}$.

To prove this theorem we rely upon the following key technical result.

\begin{lem}
\label{lem:key_technical} Let $\mathbb{P}$ denote the set
$\{M_TU_kU_k' + V_{n-k}PU_{n-k}': P \in
\mathbb{O}_{n-k}\}$. Then $\mathbb{M}(X_q,Y_q)$ $=$ $\mathbb{P}.$
\end{lem}

\noindent {\bf Proof:} Let $\mathbb{M}(U_k,M_TU_k)$ denote the set of all
$M \in \mathbb{O}_n$ such that $MU_k=M_TU_k$.  First we show that
$\mathbb{M}(X_q,Y_q)$ $=$ $\mathbb{M}(U_k,M_TU_k)$.  Since $Col(X_q)$
$=$ $Col(U_k)$, then there exists $k \times p$ matrix $A$ such that
$U_kA = X_q$.  Since $A$ has $k$ columns, then $rank(A) \leq k$.  Furthermore,
\citep[pg. 201]{S:1986} implies that $k$ $=$ $rank(U_kA)$ $\leq$
$\min\{k,rank(A)\}$, thus, $rank(A) = k$.
Therefore, from \citep[pg. 90]{S:1986}, $A$ has a right inverse.

For any $M \in \mathbb{O}_n$, we have
\begin{eqnarray*}
M \in \mathbb{M}(X_q,Y_q) &\Leftrightarrow& MU_kA = M_TU_kA \\
&\Leftrightarrow& MU_k = M_TU_k.
\end{eqnarray*}

\noindent The last $\Leftrightarrow$ follows from the fact that $A$ has a
right inverse.  We conclude that $\mathbb{M}(X_q,Y_q)$ $=$
$\mathbb{M}(U_k,M_TU_k)$.
Now we complete the proof by showing that $\mathbb{M}(U_k,M_TU_k)$
$=$ $\mathbb{P}$.

(1) For any $M \in \mathbb{P}$, there exists $P \in \mathbb{O}_{n-k}$ such that
$M$ $=$ $\{M_TU_kU_k' + V_{n-k}PU_{n-k}'\}$.  We have then

\begin{eqnarray*}
MU_k &=& M_TU_kU_k'U_k + V_{n-k}PU_{n-k}'U_k \\
&=& M_TU_k.
\end{eqnarray*}

\noindent If we can show that $M$ is orthogonal,
then $M$ $\in$ $\mathbb{M}(U_k,M_TU_k)$, so, $\mathbb{P}$ $\subseteq$
$\mathbb{M}(U_k,M_TU_k)$, as desired.  Let $U$ denote $[U_k|U_{n-k}]$
(clearly $U \in \mathbb{O}_n$).  Observe

\begin{eqnarray*}
M'M &=& U_kU_k'M_T'M_TU_kU_k' + U_kU_k'M_T'V_{n-k}PU_{n-k}' \\
&+& U_{n-k}P'V_{n-k}'M_TU_kU_k' + U_{n-k}P'V_{n-k}'M_TU_{n-k}PU_{n-k}' \\
&=&  U_kU_k' + 0 + 0 + U_{n-k}U_{n-k}' \\
&=& UU' = I_n.
\end{eqnarray*}

\noindent where the first zero in the second equality is due to the fact
that $Col(M_TU_k)$ $=$ $Col(Y_q)$, so, $V_{n-k}'M_TU_k$ $=$ $0$.

(2) Now consider $M$ $\in$ $\mathbb{M}(U_k,M_TU_k)$. It can
be shown that $Col(V_{n-k})$ $=$ $Col(MU_{n-k})$.\footnote{Since
$(MU_{n-k})'MU_k$ $=$ $0$, then $Col(MU_{n-k})$ $=$ $Col_{\bot}(MU_k)$.
Since $MU_k$ $=$ $M_TU_k$ and $Col(M_TU_k)$ $=$ $Col(Y_q)$, then it follows that
$Col_{\bot}(MU_k)$ $=$ $Col_{\bot}(M_TU_k)$ $=$ $Col_{\bot}(Y_q)$ $=$
$Col(V_{n-k})$.}  Thus, there exists
$(n-k) \times (n-k)$ matrix
$P$ with $V_{n-k}P$ $=$ $MU_{n-k}$.  Observe that
\begin{eqnarray*}
P'P &=& P'(V_{n-k}'V_{n-k})P \\
&=& (V_{n-k}P)'(V_{n-k}P) \\
&=& (MU_{n-k})'(MU_{n-k}) = I_{n-k}.
\end{eqnarray*}

\noindent Thus, $P$ $\in$ $\mathbb{O}_{n-k}$.  Moreover,
\begin{eqnarray*}
MU &=& M[U_k|U_{n-k}] \\
&=& [M_TU_k|MU_{n-k}] \\
&=& [M_TU_k|V_{n-k}P].
\end{eqnarray*}

\noindent Thus,
\begin{eqnarray*}
M &=& [M_TU_k|V_{n-k}P]\left[\begin{array}{c} U_k' \\ U_{n-k}' \end{array} \right] \\
&=& M_TU_kU_k' + V_{n-k}PU_{n-k}'.
\end{eqnarray*}

\noindent Therefore, $M$ $\in$ $\mathbb{P}$, so,
$\mathbb{M}(U_k,M_TU_k)$ $\subseteq$ $\mathbb{P}$, as desired.
\qed

Now we prove Theorem \ref{thm:IOkey}.

\noindent {\bf Proof:} Clearly $L$ is an affine map.  Moreover, Lemma \ref{lem:key_technical} directly
implies that $L$ maps $\mathbb{O}_{n-k}$ onto $\mathbb{M}(X_q,Y_q)$.  To see
that $L$ is one-to-one, consider $P_1, P_2 \in \mathbb{O}_{n-k}$ such that
$L(P_1) = L(P_2)$.  By definition, $M_TU_kU_k'+V_{n-k}P_1U_{n-k}'$ $=$
$M_TU_kU_k'+V_{n-k}P_2U_{n-k}'$, thus, $V_{n-k}P_1U_{n-k}'$ $=$
$V_{n-k}P_2U_{n-k}'$.  Therefore $P_1$ $=$
$V_{n-k}'V_{n-k}P_1U_{n-k}'U_{n-k}$ $=$ $V_{n-k}'V_{n-k}P_2U_{n-k}'U_{n-k}$
$=$ $P_2$.

To complete the proof, consider $P \in \mathbb{O}_{n-k}$.  We have,
$V_{n-k}'L(P)U_{n-k}$ $=$
$V_{n-k}'M_TU_kU_k'U_{n-k} + V_{n-k}'V_{n-k}PU_{n-k}'U_{n-k}$ $=$ $0+P$.
Moreover, consider $M \in \mathbb{M}(X_q,Y_q)$.
By Lemma \ref{lem:key_technical},
there exists $P_M \in \mathbb{O}_{n-k}$ such that $M$ $=$
$M_TU_kU_k'+V_{n-k}P_MU_{n-k}'$.  We have $L(V_{n-k}'MU_{n-k})$ $=$
$L(P_M)$ $=$ $M$.  Therefore, the inverse of $L$ is
$M \in \mathbb{M}(X_q,Y_q)$ $\mapsto$ $V_{n-k}'MU_{n-k}$.
\qed

\subsection{Known Input Attack: A Rigorous Development of the Closed-Form Expression for $\rho(x_{\hat{j}},\epsilon)$}

Up to (\ref{IO_derivation1}),
we had derived the following result (for $P$ chosen uniformly from $\mathbb{O}_{n-k}$):

\begin{equation}
\label{eq:apendix1}
\rho(x_{\hat{j}},\epsilon) = Pr(||P'B'(U_{n-k}'x_{\hat{j}}) - (U_{n-k}'x_{\hat{j}})|| \leq ||x_{\hat{j}}||\epsilon),
\end{equation}

\noindent where $B \in \mathbb{O}_{n-k}$ and satisfies
$M_TU_{n-k}B = V_{n-k}$.  Now we provide a rigorous proof of
(\ref{IO_derivation2}), {\em i.e.} the r.h.s. above equals
$Pr(||P'(U_{n-k}'x_{\hat{j}}) - (U_{n-k}'x_{\hat{j}})|| \leq
||x_{\hat{j}}||\epsilon)$.  To do so, we need some material from
measure theory.

Because $\mathbb{O}_{n-k}$ is a locally compact topological group
\citep[pg. 293]{Artin_91}, it has a Haar probability measure,
denoted by $\mu$, over $\mathbb{B}$, the Borel algebra on
$\mathbb{O}_{n-k}$. This is commonly
regarded as the standard uniform probability measure over
$\mathbb{O}_{n-k}$.
Its key property is {\em left-invariance}: for
all $\mathcal{B} \in \mathbb{B}$ and all $M \in \mathbb{O}_{n-k}$,
$\mu(\mathcal{B})$ $=$ $\mu(M\mathcal{B})$, {\em i.e.,} shifting
$\mathcal{B}$ by a rigid motion does not change its probability
assignment.

%%%%%%%%%%%%%%%%%%%%%%
%$z_{\hat{j}}$ denote
%  U_{n-k}'x_{\hat{j}} and
%$\epsilon_{\hat{j}}$ denote
%   ||x_{\hat{j}}||\epsilon
%%%%%%%%%%%%%%%%%%%%%

Let
$\mathbb{O}_{n-k}(U_{n-k}'x_{\hat{j}},||x_{\hat{j}}||\epsilon)$ denote the set of
all $P \in \mathbb{O}_{n-k}$ such that
$||P'(U_{n-k}'x_{\hat{j}})-(U_{n-k}'x_{\hat{j}})||\leq ||x_{\hat{j}}||\epsilon)$.  Let
$\mathbb{O}^{B'}_{n-k}(U_{n-k}'x_{\hat{j}},||x_{\hat{j}}||\epsilon)$ denote the set of
all $P \in \mathbb{O}_{n-k}$ such that
$||P'B'(U_{n-k}'x_{\hat{j}})-(U_{n-k}'x_{\hat{j}})||\leq ||x_{\hat{j}}||\epsilon$.\footnote{Since
$\mathbb{O}_{n-k}(U_{n-k}'x_{\hat{j}},||x_{\hat{j}}||\epsilon)$ and
$\mathbb{O}^{B'}_{n-k}(U_{n-k}'x_{\hat{j}},||x_{\hat{j}}||\epsilon)$ are topologically closed
sets, then they are Borel subsets of $\mathbb{O}_{n-k}$, therefore, $\mu$
is defined on each of these.}  By definition of $\mu$
we have,

\begin{eqnarray*}
\mu(\mathbb{O}_{n-k}(U_{n-k}'x_{\hat{j}},||x_{\hat{j}}||\epsilon)) &=& Pr(P \mbox{ uniformly chosen from } \mathbb{O}_{n-k} \mbox{ lies in } \mathbb{O}_{n-k}(U_{n-k}'x_{\hat{j}},||x_{\hat{j}}||\epsilon)) \\
&=&  Pr(||P'(U_{n-k}'x_{\hat{j}})  - (U_{n-k}'x_{\hat{j}})|| \leq ||x_{\hat{j}}||\epsilon),
\end{eqnarray*}

\noindent and,

\begin{eqnarray*}
\mu(\mathbb{O}^{B'}_{n-k}(U_{n-k}'x_{\hat{j}},||x_{\hat{j}}||\epsilon)) &=& Pr(P \mbox{ uniformly chosen from } \mathbb{O}_{n-k} \mbox{ lies in } \mathbb{O}^{B'}_{n-k}(U_{n-k}'x_{\hat{j}},||x_{\hat{j}}||\epsilon)) \\
&=&  Pr(||P'B'(U_{n-k}'x_{\hat{j}})  - (U_{n-k}'x_{\hat{j}})|| \leq ||x_{\hat{j}}||\epsilon),
\end{eqnarray*}

\noindent Therefore,

\begin{eqnarray}
Pr(||P'B'(U_{n-k}'x_{\hat{j}}) - (U_{n-k}'x_{\hat{j}})|| \leq ||x_{\hat{j}}||\epsilon) &=& \mu(\mathbb{O}^{B'}_{n-k}(U_{n-k}'x_{\hat{j}},||x_{\hat{j}}||\epsilon)) \nonumber \\
&=& \mu(B\mathbb{O}^{B'}_{n-k}(U_{n-k}'x_{\hat{j}},||x_{\hat{j}}||\epsilon)) \nonumber \\
&=& \mu(\mathbb{O}_{n-k}(U_{n-k}'x_{\hat{j}},||x_{\hat{j}}||\epsilon)) \label{eq:appendix2} \\
&=& Pr(||P'(U_{n-k}'x_{\hat{j}}) - (U_{n-k}'x_{\hat{j}})|| \leq ||x_{\hat{j}}||\epsilon) \nonumber
\end{eqnarray}

\noindent where the second equality is due to the left-invariance of $\mu$ and
the third equality is due to the fact that
$B\mathbb{O}^{B'}_{n-k}(U_{n-k}'x_{\hat{j}},||x_{\hat{j}}||\epsilon)$ can be shown to
equal $\mathbb{O}_{n-k}(U_{n-k}'x_{\hat{j}},||x_{\hat{j}}||\epsilon)$.

Since the last equality above was for intuitive purposes only, we will ignore it
in completing the derivation of a closed form expression.
(\ref{eq:apendix1}) and (\ref{eq:appendix2}) imply

$$\rho(x_{\hat{j}},\epsilon) = \mu(\mathbb{O}_{n-k}(U_{n-k}'x_{\hat{j}},||x_{\hat{j}}||\epsilon)).$$

\noindent Recall that $S_{n-k}(||U_{n-k}'x_{\hat{j}}||)$ denotes the
hyper-sphere in $\Re^{n-k}$ with radius $||U_{n-k}'x_{\hat{j}}||$ and centered at the
origin and $S_{n-k}(U_{n-k}'x_{\hat{j}},||x_{\hat{j}}||\epsilon)$
denotes the points contained by $S_{n-k}(||U_{n-k}'x_{\hat{j}}||)$ whose
distance from $U_{n-k}'x_{\hat{j}}$ is no greater than
$||x_{\hat{j}}||\epsilon$.  Using basic principles from
measure theory, it can be shown that\footnote{$S_1(||U_{1}'x_{\hat{j}}||)$ consists of two points.
Recall that we define $\frac{SA(S_{1}(U_{1}'x_{\hat{j}},||x_{\hat{j}}||\epsilon))}{SA(S_{1}(||U_{1}'x_{\hat{j}}||))}$
as 0.5 if $S_1(U_{1}'x_{\hat{j}},||x_{\hat{j}}||\epsilon)$ is one point, and as 1 otherwise.  Moreover,
we define $\frac{SA(S_{0}(U_{0}'x_{\hat{j}},||x_{\hat{j}}||\epsilon))}{SA(S_{0}(||U_{0}'x_{\hat{j}}||))}$ as 1.}

$$\mu(\mathbb{O}_{n-k}(U_{n-k}'x_{\hat{j}},||x_{\hat{j}}||\epsilon)) = \frac{SA(S_{n-k}(U_{n-k}'x_{\hat{j}},||x_{\hat{j}}||\epsilon))}{SA(S_{n-k}(||U_{n-k}'x_{\hat{j}}||))}$$

\noindent We have arrived at
Equation (\ref{eq:intuitive_closed2}) from Section \ref{sec:closed}.
Next, we derive the desired closed-form
expression (\ref{eq:closed1}).  To simplify exposition, we prove the following result
for $m \geq 0$, $z \in \Re^m$, and $c \geq 0$ (by plugging in
$m = n-k$, $z$ $=$ $U_{n-k}'x_{\hat{j}}$, and
$c$ $=$ $||x_{\hat{j}}||\epsilon$, (\ref{eq:closed1}) follows).

{\footnotesize
\begin{equation}\label{eq:appendix_closed}
\frac{SA(S_{m}(z,c))}{SA(S_{m}(||z||))} = \left\{ \begin{array}{ll}
1 & \mbox{if $m=0$;} \\
1 & \mbox{if $c \geq ||z||2$ and $m \geq 1$;} \\
0.5 & \mbox{if $c < ||z||2$ and $m = 1$;} \\
1 - (1/\pi)arccos([c/(||z||\sqrt{2})]^2-1) & \mbox{if $||z||\sqrt{2} < c < ||z||2$ and $m = 2$;} \\
1 - \frac{(m-1)\Gamma([m+2]/2)}{m\sqrt{\pi}\Gamma([m+1]/2)}\int_{\theta_1=0}^{arccos([c/(||z||\sqrt{2})]^2-1)}sin^{m-1}(\theta_1)\,d\theta_1  & \mbox{if $||z||\sqrt{2} < c < ||z||2$ and $m \geq 3$;} \\
(1/\pi)arccos(1-[c/(||z||\sqrt{2})]^2) & \mbox{if $c \leq ||z||\sqrt{2}$ and $m = 2$;} \\
\frac{(m-1)\Gamma([m+2]/2)}{m\sqrt{\pi}\Gamma([m+1]/2)}\int_{\theta_1=0}^{arccos(1-[c/(||z||\sqrt{2})]^2)}sin^{m-1}(\theta_1)\,d\theta_1 & \mbox{if $c \leq ||z||\sqrt{2}$ and $m \geq 3$.} \end{array} \right.
\end{equation}
}

\noindent Before proving
(\ref{eq:appendix_closed}) we establish:

\begin{equation}\label{eq:surface_area}
\mbox{For $b \geq 2$ and $r > 0$, } SA(S_b(r)) = \frac{br^{b-1}\pi^{b/2}}{\Gamma((b+2)/2)}.
\end{equation}

\noindent Indeed, with $Vol(.)$ denoting volume, it can be shown that
$SA(S_b(r))$ $=$ $\frac{\,dVol(S_b(r))}{\,dr}$ $=$ $Vol(S_b(1))\frac{\,dr_b}{\,dr}$
$=$ $\frac{\pi^{b/2}br^{b-1}}{\Gamma((b+2)/2)}$.  The last equality follows from \citep{HH:2008}.
Now we return to proving (\ref{eq:appendix_closed}).

If $m=0$, then the surface area ratio equals 1 by definition.
If $c \geq ||z||2$ and $m \geq 1$, then the ratio equals 1
since $S_{m}(z,c)$ $=$ $S_{m}(||z||)$.
If $c < ||z||2$ and $m = 1$, then, the ratio equals 0.5
since $S_{1}(z,c)$ $= \{z\}$ and $S_{1}(||z||)$ $=\{z,-z\}$.
For the remainder of the derivation, we assume that $m \geq 2$ and, without loss of
generality, $z$ is at the ``north pole'' of the hyper-sphere
$S_{m}(||z||)$, {\em i.e.} $z$ $=$ $(1,0,0,\cdots,0)$.

{\bf Case} $c \leq ||z||\sqrt{2}$: The set of points on $S_m(||z||)$
whose distance from
$z$ equals $c$ is the intersection of $S_m(||z||)$ with the hyper-plane whose
perpendicular to $z$ is of length $h$ as seen in Figure \ref{fig:cap1}.  Thus,
$S_{m}(z,c)$ are all those points on $S_m(||z||)$ not below
that hyper-plane.

%
%It can be shown that $h$ $=$ $\frac{(||x_{\hat{j}||\epsilon)^2}{2||U_{n-k}'x_{\hat{j}}||}$ (hence lies
%between 0 and $||U_{n-k}'x_{\hat{j}}||$).
%

\begin{figure}[t]
\begin{minipage}[t]{0.49\linewidth}
\centering
\includegraphics[width=2.5in]{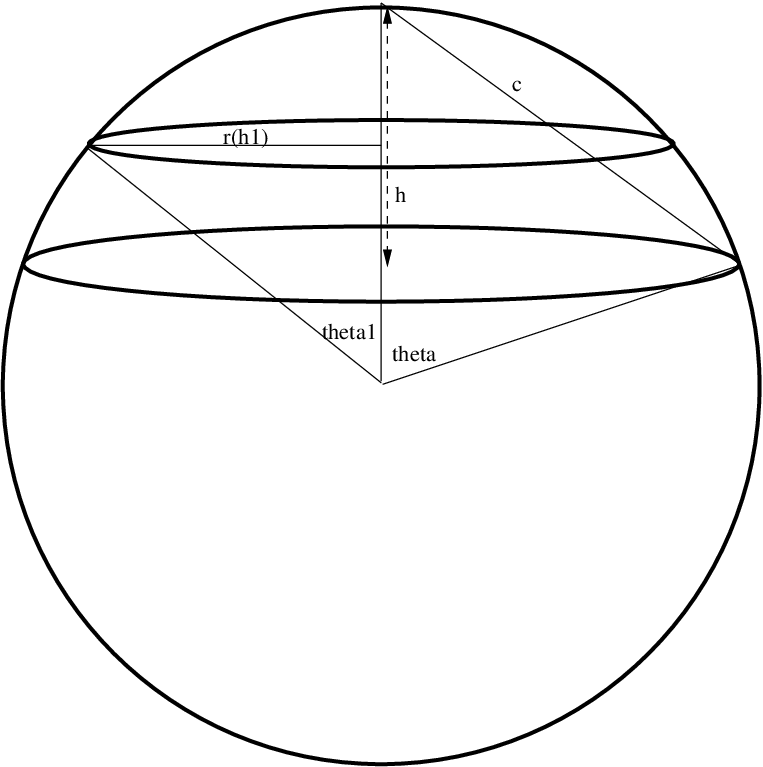}
\caption{The hyper-sphere $S_m(||z||)$ and two ``north pole'' caps ($c \leq ||z||\sqrt{2}$).}
\label{fig:cap1}
\end{minipage}%
\hspace{0.03\linewidth}%
\begin{minipage}[t]{0.48\linewidth}
\centering
\includegraphics[width=2.5in]{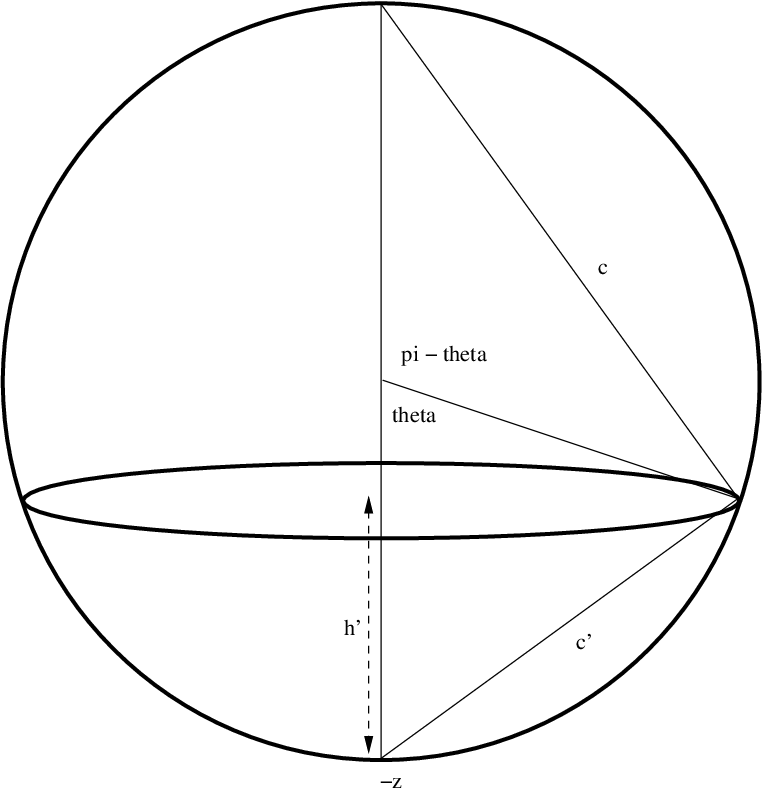}
\caption{The hyper-sphere $S_m(||z||)$ and one ``south pole'' cap ($||z||\sqrt{2} < c < ||z||2$).}
\label{fig:cap2}
\end{minipage}
\end{figure}

{\bf Sub-case} $m=2$: Since $S_2(||z||)$ is an ordinary circle, then the angle
$\theta$ in
Figure \ref{fig:cap1} determines the surface area ratio as follows
$\frac{(2\theta/2\pi)SA(S_{2}(||z||))}{SA(S_{2}(||z||))}$ $=$
$\theta/\pi$.  Moreover, since $\theta$ is the top angle of an isosceles triangle
with sides of length $||z||$ and base of length $c$, then $sin(\theta/2)$ $=$ $c/(2||z||)$.
The half-angle formula implies that $\theta$ $=$ $arccos(1 - [c/(||z||\sqrt{2})]^2)$.  Therefore,
as desired,

\begin{equation}\label{sub-case1}
\frac{SA(S_{2}(z,c))}{SA(S_{2}(||z||))} = (1/\pi)arccos(1 - [c/(||z||\sqrt{2})]^2).
\end{equation}

{\bf Sub-case} $m \geq 3$: Here, computing the surface area ratio is more complicated and requires
an appeal to the integral definition of the cap surface area.  Consider the intersection of
$S_m(||z||)$ with the hyper-plane
whose perpendicular to $z$ is of length $0 \leq h_1 \leq h$ as seen in Figure \ref{fig:cap1}.
The surface area of this intersection equals the surface area of $S_{m-1}(r(h_1))$.  Thus,
(\ref{eq:surface_area}) implies

\begin{eqnarray*}
SA(S_{m}(z,c)) &=& \int_{h_1=0}^{h}SA(S_{m-1}(r(h_1)))\,dh_1 \\
&=& \left(\frac{(m-1)\pi^{(m-1)/2}}{\Gamma((m+1)/2)} \right)\int_{h_1=0}^{h}(r(h_1))^{m-2} \,dh_1.
\end{eqnarray*}

\noindent To evaluate the integral, we change coordinates with $h_1$ $=$
$||z||(1-cos(\theta_1))$.  So, $h_1$ $=$ $0,h$ implies that $\theta_1$ $=$
$0,arccos(1 - h/||z||)$.  And, $r(||z||(1-cos(\theta_1))$ $=$
$||z||sin(\theta_1)$ $=$ $\frac{\,dh_1}{\,d\theta_1}$.  Therefore,

\begin{eqnarray*}
\int_{h_1=0}^{h}(r(h_1))^{m-2} \,dh_1 &=& \int_{\theta_1=0}^{arccos(1-h/||z||)}r(||z||(1-cos(\theta_1))^{m-2}\frac{\,dh_1}{\,d\theta_1}\,d\theta_1 \\
&=& \int_{\theta_1=0}^{arccos(1-h/||z||)}||z||^{m-2}sin^{m-2}(\theta_1)||z||sin(\theta_1)\,d\theta_1 \\
&=& ||z||^{m-1}\int_{\theta_1=0}^{arccos(1-h/||z||)}sin^{m-1}(\theta_1)\,d\theta_1.
\end{eqnarray*}

\noindent Plugging this into the previous equations for $SA(S_m(z,c))$ and
using (\ref{eq:surface_area}), we get

\begin{eqnarray*}
\frac{SA(S_{m}(z,c))}{SA(S_{m}(||z||))} &=& \left(\frac{(m-1)\pi^{(m-1)/2}||z||^{m-1}\Gamma((m+2)/2)}{\Gamma((m+1)/2)m||z||^{m-1}\pi^{m/2}} \right)\int_{\theta_1=0}^{arccos(1-h/||z||)}sin^{m-1}(\theta_1)\,d\theta_1 \\
&=& \left(\frac{(m-1)\Gamma((m+2)/2)}{\Gamma((m+1)/2)m\sqrt{\pi}} \right)\int_{\theta_1=0}^{arccos(1-h/||z||)}sin^{m-1}(\theta_1)\,d\theta_1.
\end{eqnarray*}

\noindent Since $h$ $=$ $\frac{c^2}{2||z||}$, then, as desired, we get

\begin{equation} \label{sub-case2}
\frac{SA(S_{m}(z,c))}{SA(S_{m}(||z||))} = \left(\frac{(m-1)\Gamma((m+2)/2)}{\Gamma((m+1)/2)m\sqrt{\pi}} \right)\int_{\theta_1=0}^{arccos(1-[c/||z||\sqrt{2}]^2)}sin^{m-1}(\theta_1)\,d\theta_1.
\end{equation}

{\bf Case} $||z||\sqrt{2} < c < ||z||2$:  As depicted in Figure \ref{fig:cap2},
$S_m(z,c)$ contains the entire northern hemisphere of $S_m(||z||)$.  Let
$S_n(-z,c)$ denote the ``south pole'' cap defined by $h'$ (and $c'$) in Figure
\ref{fig:cap2} (clearly $c'$ $\leq$ $||z||\sqrt{2}$).  We have

\begin{equation}\label{case1.5}
\frac{SA(S_m(z,c))}{SA(S_m(||z||))} = 1-\frac{SA(S_m(-z,c')}{SA(S_m(||z||))}.
\end{equation}

\noindent By replacing ``$c$'' with ``$c'$'' in (\ref{sub-case1}) and (\ref{sub-case2}) then plugging
the resulting expression into (\ref{case1.5}) we get,

{\footnotesize
\begin{equation}\label{case2}
\frac{SA(S_{m}(z,c))}{SA(S_{m}(||z||))} = \left\{ \begin{array}{ll}
1 - (1/\pi)arccos(1-[c'/(||z||\sqrt{2})]^2) & \mbox{if $m = 2$;} \\
1 - \frac{(m-1)\Gamma([m+2]/2)}{m\sqrt{\pi}\Gamma([m+1]/2)}\int_{\theta_1=0}^{arccos(1-[c'/(||z||\sqrt{2})]^2)}sin^{m-1}(\theta_1)\,d\theta_1 & \mbox{if $m \geq 3$.} \end{array} \right.
\end{equation}
}

\noindent From Figure \ref{fig:cap2}, it can be seen that $\theta$ is the top angle on an isosceles
triangle with sides of length $||z||$ and base of length $c'$.  So, $sin(\theta/2)$ $=$ $\frac{c'}{2||z||}$.
The half-angle formula implies $cos(\theta)$ $=$ $1 - [c'/(||z||\sqrt{2})]^2$.
Similar reasoning shows $cos(\pi-\theta)$ $=$ $1 - [c/(||z||\sqrt{2})]^2$.  Since $0\leq \theta \leq \pi/2$, then
$cos(\pi-\theta)$ $=$ $-cos(\theta)$.  Thus, $[c/(||z||\sqrt{2})]^2-1$ $=$ $1 - [c'/(||z||\sqrt{2})]^2$.  Plugging
$2-[\frac{c}{||z||\sqrt{2}}]^2$ in for $[\frac{c'}{||z||\sqrt{2}}]^2$ in (\ref{case2})
yields the desired results.

\subsection{Known Input Attack: Computing the Closed-Form Expression for $\rho(x_{\hat{j}},\epsilon)$}

Next we develop recursive procedures for computing (\ref{eq:closed2}).
This amounts to computing the following two functions:
(i) $GR(m)$ $=$ $\Gamma([m+2]/2)/\Gamma([m+1]/2)$ for $m \geq 1$; (ii)
$SI(z,m)$ $=$ $\int_{\theta_1 = 0}^{arccos(z)}sin^{m-1}(\theta_1)\,d\theta_1$
for $1 \geq z \geq 0$ and $m \geq 1$.
Indeed, (\ref{eq:closed2}) is equivalent to

{\footnotesize
\begin{equation}
\rho(x_{\hat{j}},\epsilon) = \left\{ \begin{array}{ll}
1 & \mbox{if $n-k=0$;} \\
1 & \mbox{if $||y_j||\epsilon \geq ||V_{n-k}'y_j||2$ and $n-k \geq 1$;} \\
0.5 & \mbox{if $||y_j||\epsilon < ||V_{n-k}'y_j||2$ and $n-k = 1$;} \\
1 - (1/\pi)arccos\left(\left[\frac{||y_j||\epsilon}{||V_{n-k}'y_j||\sqrt{2}} \right]^2 - 1\right)& \mbox{if $||V_{n-k}'y_j||\sqrt{2} < ||y_j||\epsilon < ||V_{n-k}'y_j||2$ and $n-k = 2$;} \\
1 - \frac{(n-k-1)GR(n-k)}{(n-k)\sqrt{\pi}}SI\left(\left[\frac{||y_j||\epsilon}{||V_{n-k}'y_j||\sqrt{2}} \right]^2 - 1,n-k\right) & \mbox{if $||V_{n-k}'y_j||\sqrt{2} < ||y_j||\epsilon < ||V_{n-k}'y_j||2$ and $n-k \geq 3$;} \\
(1/\pi)arccos\left(1-\left[\frac{||y_j||\epsilon}{||V_{n-k}'y_j||\sqrt{2}} \right]^2\right)& \mbox{if $||y_j||\epsilon \leq ||V_{n-k}'y_j||\sqrt{2}$ and $n-k = 2$;} \\
\frac{(n-k-1)GR(n-k)}{(n-k)\sqrt{\pi}}SI\left(1-\left[\frac{||y_j||\epsilon}{||V_{n-k}'y_j||\sqrt{2}} \right]^2,n-k\right)& \mbox{if $||y_j||\epsilon \leq ||V_{n-k}'y_j||\sqrt{2}$ and $n-k \geq 3$.} \end{array} \right.
\end{equation}
}

To compute $GR(m)$ for $m \geq 1$, we use the following facts:
$\Gamma(z+1)$ $=$ $z\Gamma(z)$ for $z > 0$, $\Gamma(1/2)$ $=$ $\sqrt{\pi}$, and
$\Gamma(1)$ $= 1.$  Thus, we get
a recursive procedure for computing $GR(m)$.

{\footnotesize
\begin{equation}
GR(m) = \left\{ \begin{array}{ll}
\frac{\sqrt{\pi}}{2} & \mbox{if $m=1$;} \\
\frac{2}{\sqrt{\pi}} & \mbox{if $m=2$;} \\
\left(\frac{m}{m-1}\right)GR(m-2) & \mbox{if $m \geq 3$.} \end{array} \right.
\end{equation}
}

\noindent To compute $SI(z,m)$ for $1 \geq z \geq 0$ and $m \geq 1$, we use the following facts.
$sin^{m-2}(arccos(z))$ $=$ $[1-z^2]^{(m-2)/2}$ if $m \geq 3$.  And, $SI(z,m)$ $=$
$\left[\int sin^{m-1}(\theta_1)\,d \theta_1 \right](arccos(z))$ $-$
$\left[\int sin^{m-1}(\theta_1)\,d \theta_1 \right](0)$.  And,

{\footnotesize
\begin{equation}
\left[\int sin^{m-1}(\theta_1)\,d \theta_1\right](w) = \left\{ \begin{array}{ll}
w & \mbox{if $m-1=0$;} \\
-cos(w) & \mbox{if $m-1=1$;} \\
\frac{m-2}{m-1}\left[\int sin^{m-3}(\theta_1)\,d \theta_1\right](w) - \frac{sin^{m-2}(w)cos(w)}{m-1} & \mbox{if $m-1 \geq 2$;} \end{array} \right.
\end{equation}
}

\noindent Therefore,

{\footnotesize
\begin{equation}
SI(z,m) = \left\{ \begin{array}{ll}
arccos(z) & \mbox{if $m=1$;} \\
1-z & \mbox{if $m=2$;} \\
\frac{m-2}{m-1}SI(z,m-2) - \frac{z[1-z^2]^{(m-2)/2}}{m-1}& \mbox{if $m \geq 3$;} \end{array} \right.
\end{equation}
}

%%%%%%%%%%%%%%%%%%%%%%%%%%%%%%%%%%%%%%%%%%%%%%%%%%%%%%%%%%%%%%
% biography section
%
% If you had an eps/pdf photo file (graphicx package needed)
% the extra braces prevent the LaTeX parser from getting confused
% when it sees the complicated \includegraphics command within an
% optional argument. You can create your own macro to make things
% simpler here.
%\begin{biography}[{\includegraphics[width=1in,height=1.25in,clip,keepaspectratio]{mshell.eps}}]{Michael Shell}
% or if you just want to reserve a space for a photo:

\parpic{\includegraphics[width=1in,height=1.25in,clip,keepaspectratio]{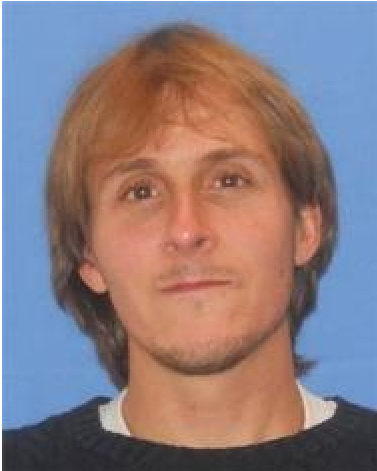}}
\noindent {\bf Chris Giannella} is an artificial intelligence engineer with the MITRE corporation in Annapolis Junction Maryland.  
His current research interests include machine learning and natural language processing.  
Prior to that, he held faculty positions at New Mexico State University, Loyola University in Maryland, and Goucher College. 
Prior to that he was a postdoctoral research associate at the University of Maryland, Baltimore County and 
completed his Ph.D. in Computer Science at Indiana University, Bloomington, Indiana in 2004.      

\parpic{\includegraphics[width=1in,height=1.25in,clip,keepaspectratio]{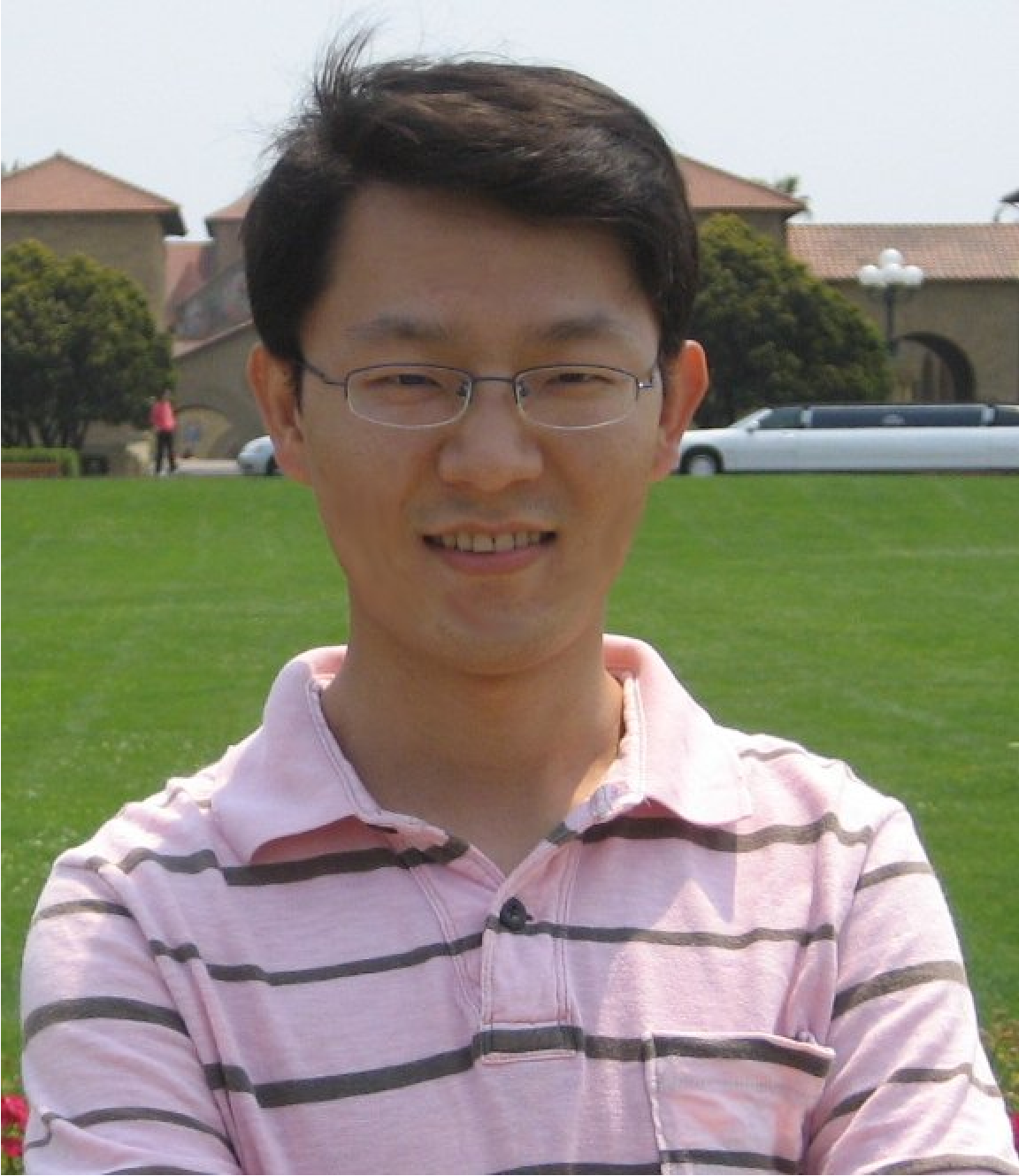}}
\noindent {\bf Kun Liu}, Ph.D. is working at LinkedIn as a Staff Software Engineer and
Applied Researcher. He is primarily focusing on user intent and
interest modeling for personalization and ad targeting. Prior to that,
he was a Scientist at Yahoo! Labs, leading several successful research
projects such as social targeting and commercial mail monetization.
Before joining Yahoo, he was a Postdoctoral Researcher at IBM Almaden
Research Center, working on privacy-preserving social-network analysis
and text analytics.  Dr. Liu received his Ph.D. in Computer Science
from University of Maryland Baltimore County. His research interests
include behavioral ad targeting, privacy-preserving data mining and
social-network analysis. He has co-authored over 25 peer-reviewed
research papers and book chapters. He also regularly serves on the
program committee of many data mining conferences ({\em e.g.}, KDD, ICDM,
PKDD, PAKDD), and as a reviewer of many scientific journals ({\em e.g.},
IEEE TKDE, ACM TKDD).

\parpic{\includegraphics[width=1in,height=1.25in,angle=-90,clip,keepaspectratio]{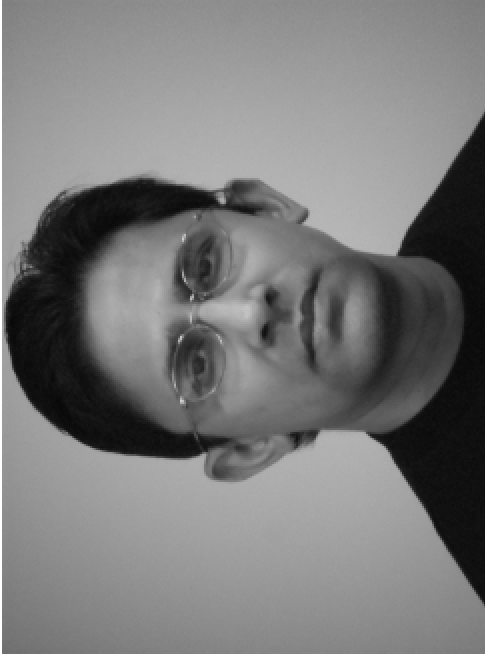}}
\noindent {\bf Hillol Kargupta} is a Professor of Computer Science at the University 
of Maryland, Baltimore County. He is also a co-founder of AGNIK, a 
vehicle performance data analytics company for mobile, distributed, and 
embedded environments. He received his Ph.D. in Computer Science from 
University of Illinois at Urbana-Champaign in 1996. His research 
interests include mobile and distributed data mining. Dr. Kargupta is an 
IEEE Fellow. He won the IBM Innovation Award in 2008 and a National 
Science Foundation CAREER award in 2001 for his research on ubiquitous 
and distributed data mining. He and his team received the 2010 Frost and 
Sullivan Enabling Technology of the Year Award for the MineFleet vehicle 
performance data mining product and the IEEE Top-10 Data Mining Case 
Studies Award. His other awards include the best paper award for the 
2003 IEEE International Conference on Data Mining for a paper on 
privacy-preserving data mining, the 2000 TRW Foundation Award, and the 
1997 Los Alamos Award for Outstanding Technical Achievement. His 
dissertation earned him the 1996 Society for Industrial and Applied 
Mathematics annual best student paper prize.He has published more than 
one hundred peer-reviewed articles. His research has been funded by the 
US National Science Foundation, US Air Force, Department of Homeland 
Security, NASA and various other organizations. He has co-edited several 
books. He serve(s/d) as an associate editor of the IEEE Transactions on 
Knowledge and Data Engineering, IEEE Transactions on Systems, Man, and 
Cybernetics, Part B and Statistical Analysis and Data Mining Journal. He 
is/was the Program Co-Chair of 2009 IEEE International Data Mining 
Conference, General Chair of 2007 NSF Next Generation Data Mining 
Symposium, Program Co-Chair of 2005 SIAM Data Mining Conference and 
Associate General Chair of the 2003 ACM SIGKDD Conference, among others.

% You can push biographies down or up by placing
% a \vfill before or after them. The appropriate
% use of \vfill depends on what kind of text is
% on the last page and whether or not the columns
% are being equalized.

%\vfill

% Can be used to pull up biographies so that the bottom of the last one
% is flush with the other column.
%\enlargethispage{-5in}

% that's all folks
\end{document}